\RequirePackage{fix-cm}

\documentclass[natbib,smallextended]{svjour3}       
\usepackage[english]{babel}
\usepackage[utf8x]{inputenc}
\usepackage[T1]{fontenc}
\usepackage{makecell}
\usepackage[colorinlistoftodos]{todonotes} 

\usepackage[a4paper,top=3cm,bottom=2cm,left=3cm,right=3cm,marginparwidth=1.75cm]{geometry}
\usepackage[parfill]{parskip}

\usepackage{graphicx}
\usepackage[]{algorithm2e}
\usepackage[colorinlistoftodos]{todonotes}
\usepackage[colorlinks=true, allcolors=blue]{hyperref}
\usepackage{mathtools}

\newcommand{\cS}{\mathcal{S}}
\newcommand{\cB}{\mathcal{B}}

\DeclarePairedDelimiter\ceil{\lceil}{\rceil}

\newcommand{\totHomScore}{\textsc{Total Homoplasy Score}}
\newcommand{\maxHomRate}{\textsc{Max Homoplasy Ratio}}
\newcommand{\maxHomScore}{\textsc{Max Homoplasy Score}}
\newcommand{\totHomRate}{\textsc{Total Homoplasy Ratio}}

\newcommand{\fewParBlo}{\textsc{Fewest Parsimonious Blocks}}

\newcommand{\contBlockS}[2]{[#1 - #2]}
\newcommand{\contBlock}[3]{#1\contBlockS{#2}{#3}}

\newcommand{\markj}[1]{\textcolor{black}{#1}} 
\newcommand{\fabio}[1]{\textcolor{black}{#1}} 
\newcommand{\cel}[1]{\textcolor{black}{#1}}

\newcommand{\steven}[1]{\textcolor{black}{#1}}
\newcommand{\philippe}[1]{\textcolor{black}{#1}}

\newcommand{\markjcut}[1]{\textcolor{black}{}}

\title{Cutting an alignment with Ockham's razor}

\author{Mark Jones  \and Philippe Gambette \and Leo van Iersel\thanks{Research funded in part by the Netherlands Organization for Scientific Research (NWO), including Vidi grant 639.072.602 and Gravitation grant NETWORKS-024.002.003, and partly by the 4TU Applied Mathematics Institute.} \and Remie Janssen \and Steven Kelk \and Fabio Pardi \and Celine Scornavacca}

\institute{
Mark Jones \at Centrum Wiskunde \& Informatica (CWI), P.O. Box 94079,
1090 GB Amsterdam,
Netherlands.
\\
\email{markelliotlloyd@gmail.com}
\and
Philippe Gambette \at Laboratoire d'Informatique Gaspard-Monge (LIGM), Universit\'e Paris-Est, CNRS, ENPC, ESIEE Paris, UPEM, F-77454, Marne-la-Vall\'ee, France. Email: philippe.gambette@u-pem.fr 
\and
 Leo van Iersel and Remie Janssen  \at Delft Institute of Applied Mathematics, Delft University of Technology, Van Mourik Broekmanweg 6,
2628 XE Delft,
The Netherlands. \email{\{L.J.J.vanIersel, R.Janssen-2\}@tudelft.nl} 
\and 
Steven Kelk \at Department of Data Science and Knowledge Engineering (DKE), Maastricht University, P.O. Box 616, 6200 MD Maastricht, Netherlands. \email{steven.kelk@maastrichtuniversity.nl}
\and
Fabio Pardi \at LIRMM, Universit\'{e} de Montpellier, CNRS, Montpellier, France. \email{pardi@lirmm.fr}
\and
Celine Scornavacca \at Institut des Sciences de l'Evolution, 
Universit\'{e} de Montpellier, CNRS, IRD, EPHE, 
34095 Montpellier Cedex 5, France. \email{celine.scornavacca@umontpellier.fr}
}

\date{Received: date / Accepted: date}

\begin{document}
\maketitle



\begin{abstract}

\cel{In this article, w}e investigate different parsimony-based approaches towards finding recombination breakpoints in a multiple sequence alignment. \philippe{This recombination detection task is crucial in order to avoid errors in 
\fabio{evolutionary analyses} caused by mixing together portions of sequences which had a different evolution history.} \steven{Following an overview of the field of recombination detection, we} formulate four computational problems for this task with different objective functions. The four problems aim to minimize (1) the total homoplasy of all blocks (2) the maximum homoplasy per block (3) the total homoplasy ratio of all blocks and (4) the maximum homoplasy ratio per block. We describe algorithm\cel{s} for each of these problems, which \cel{are} fixed-parameter tractable (FPT) when the characters are binary. We have implemented and tested the algorithms on simulated data, showing that minimizing the total homoplasy gives, in most cases, the most accurate results. \steven{Our implementation and experimental data have been made publicly available}. Finally, we also consider \cel{the problem of} combining blocks into non-contiguous blocks consisting of at most~$p$ contiguous parts. Fixing the homoplasy \cel{$h$ of each block} to~$0$, we show that this problem is NP-hard when~$p\geq 3$, but polynomial-time solvable for~$p=2$. Furthermore, the problem is FPT with parameter~$h$ for binary characters when~$p=2$. A number of interesting problems remain \cel{open}.
\end{abstract}

\keywords{Recombination breakpoints, homoplasy, parsimony, block partitioning, exact algorithms, experiments, fixed parameter tractability.}


{\bf Mathematics Subject Classification}  68Q25 92D15 92D20 

\section{Introduction}

When a multiple alignment contains sequences whose ancestors have undergone recombination, the evolutionary histories of different parts of the alignment are represented by different phylogenetic trees. 
\fabio{In this case, using standard evolutionary analysis tools that assume a single phylogenetic tree for the entire alignment can introduce important biases for several inference tasks.
For example it has been shown that in the presence of recombination, positive selection pressures~\citep{anisimova2003effect,kosakovsky2008estimating,arenas2010coalescent}, as well as rate variation among sites and lineages tends to be overestimated~\citep{schierup2000consequences,schierup2000recombination}. Recombination can also severely affect the inference of genetic distances~\citep{schierup2000consequences,lemey2009}, of ancestral sequences~\citep{arenas2010effect}, of demographic history~\citep{schierup2000consequences}, and of course of the phylogeny itself, which may bear little resemblance to the true reticulate history of the sequences~\citep{posada2002effect}.}



\fabio{An often viable solution to these problems is to 
partition the input alignment into recombination-free blocks, which can then be analysed with standard methods using a single phylogenetic tree per block
\citep[see, e.g.,][for its application to the detection of positive selection in viruses]{scheffler2006robust}. 
Inferring ``locus trees'' for a large number of recombination-free blocks has also become very popular in the context of species tree inference under the multi-species coalescent model~\citep[see][for a recent review]{xu2016challenges}, where metiotic recombination is assumed to have happened between---but not within---blocks.}


\fabio{For these reasons, a recurring preliminary step in evolutionary bioinformatics is to infer the putative locations within an alignment where recombination has occurred, that is, the \emph{recombination breakpoints}, with the goal of partitioning the alignment into blocks for further downstream analyses.}
In this paper, we closely examine a number of formulations of the problem of alignment partitioning in the presence of recombination, and algorithms that allow to efficiently solve them.
Many increasingly sophisticated methods have been proposed for this or similar tasks~\citep{salminen2009,martin2011analysing}. Below, we present an overview of the main ideas behind these methods. 

Phylogenetic incompatibilities between sites in an alignment can either be due to recurrent sub\-stitutions---that is, the same or the inverse substitution having arisen on different branches in the tree---\-or to recombination having occurred between those sites~\citep[see box 15.1 in][]{lemey2009}. Distinguishing between these two phenomena\----recurrent substitution and recombination---\-is strongly influenced by the prior belief about their relative frequencies. For example, if recurrent substitutions are assumed to be impossible (the \emph{infinite sites model}), every pair of incompatible sites must be separated by a recombination breakpoint, an observation that led to one of the earliest methods for breakpoints estimation~\citep{hudson1985statistical}. On the other hand, classical phylogenetic inference assumes no recombination, and all incompatibilities between characters are explained by recurrent substitution.  When both recombination and recurrent substitution are possible, an observation that underlies virtually all the methodology for recombination detection (and the present work is no exception) is that recombination leads to the spatial clustering along the alignment of compatible or nearly-compatible sites. In other words, in the presence of recombination, sites carrying similar phylogenetic signals tend to be closer than expected by chance alone.

The literature about recombination detection is very rich. A review from 2011 reported that there were already about 90 tools available at the time ~\citep[see Table S1 in][]{martin2011analysing}. In practice, however, these were conceived with many different tasks and goals in mind (e.g., detecting evidence of recombination, estimating breakpoints, identifying the parental sequences of the recombinants, etc.). Although it is beyond the scope of this article to provide a complete overview, there are a number of recurring ideas that are easy to describe. The methods allowing the identification of breakpoints in a multiple alignment can be roughly categorized in 3 groups~\citep{martin2011analysing}.

\emph{Similarity-} or \emph{distance-based} methods constitute the conceptually simplest approach. They inspect the variation along the alignment of the (dis-)similarity between sequences, measured in terms of percent identity or of evolutionary distance (estimated via standard substitution models). Changes along the alignment in the relative similarities among sequences are interpreted as evidence of recombination. Examples of this approach are \textsc{Rip}~\citep{siepel1995computer}, \textsc{PhilPro}~\citep{weiller1998phylogenetic}, \textsc{SimPlot}~\citep{lole1999full}, \textsc{Rat}~\citep{etherington2004recombination} and \textsc{T-Recs}~\citep{tsimpidis2017t}. A problem with these methods is that the relation between pairwise similarities and phylogenetic relatedness is not straightforward, for example because of rate variation among lineages or selection pressures, meaning that a change in relative similarity does not always indicate phylogenetic incongruence or recombination~\citep{lemey2009}.

\emph{Substitution distribution} methods focus on a small subset of sequences (e.g., a triple) and, for each such subset in the alignment, they test whether some site patterns (e.g., those with the form $xxy$, $xyx$ and $xyy$) occur in clustered locations along the alignment. Recombination breakpoints are identified with those positions where, as we move along the alignment, a change in the relative frequencies of these patterns is detected. These methods are often, but not exclusively, based on the use of a sliding window. Examples of this approach are \textsc{GeneConv}~\citep{sawyer1989statistical,padidam1999possible}, \textsc{MaxChi}~\citep{smith1992analyzing}, \textsc{Chimaera}~\citep{posada2001evaluation,martin2010rdp3}, \textsc{Siscan}~\citep{gibbs2000sister}, \textsc{3-Seq}~\citep{boni2007exact,lam2017improved}, \textsc{Rapr}~\citep{song2018tracking}. The necessity to analyse only a subset at a time, instead of all sequences simultaneously, may be problematic when the alignment contains many sequences. This is not just for computational reasons (e.g., the number of triples grows cubically in the number of sequences) and statistical reasons (correcting for the rapidly increasing multiple tests may lower the power to detect recombination~\citep{martin2017detecting}). 
Another important issue is that the situation where the alignment only contains sequences that are direct recombinants of other sequences within the alignment is in fact the exception rather than the rule. If an alignment contains several descendants of the same recombinant, or a parental sequence of a recombinant is the ancestor of several sequences in the alignment, then the same recombination event should be detected in several partially overlapping subsets. Being able to recognize when two subsets have detected the same event is a difficult and open question~\citep{song2018tracking}.

\emph{Phylogenetic-based} methods seek to detect when different contiguous blocks in the alignment support different phylogenetic trees. A seminal paper in this context is due to J.~Hein~\citep{hein1993heuristic}. That paper introduced an optimization problem whose goal is to determine a sequence of breakpoints in the alignment and a sequence of phylogenetic trees for each of the blocks delimited by the breakpoints, so as to minimize a linear combination of the parsimony scores of the trees and of the cost of breakpoints. Each breakpoint has a cost that expresses the minimum number of recombination events that are necessary to ``switch'' between 
the two trees that it separates. Solving this problem exactly is impractical for several reasons~\citep[see the formal definition and discussion in Sec.\ 9.8.1 of][]{huson2010phylogenetic}, so a heuristic named \textsc{RecPars} was proposed~\citep{hein1993heuristic,huson2010phylogenetic}. Subsequent methods relied on distance-based phylogenetics (e.g., \textsc{BootScan}~\citep{salminen1995identification}, \textsc{Topal}~\citep{mcguire1997graphical}), maximum-likelihood techniques (e.g., \textsc{Plato}~\citep{grassly1997likelihood}, \textsc{Lard}~\citep{holmes1999phylogenetic}, \textsc{Gard}~\citep{kosakovsky2006automated,kosakovsky2006gard}), while more recent methods use hidden Markov models whose hidden states correspond to phylogenetic trees, and where the observations are the columns of the alignment. These models, sometimes referred to as \emph{phylo-HMMs}, can be used to infer probable partitions of the input alignment into recombination-free blocks and are at the core of tools such as \textsc{Barce}~\citep{husmeier2003detecting}, \textsc{DualBrothers}~\citep{minin2005dual}, \textsc{StHmm}~\citep{webb2008phylogenetic} and many others~\citep{hobolth2007genomic,bloomquist2008stepbrothers,deOliveira2008phylogenetic,dutheil2009ancestral,boussau2009mixture}. Although phylogenetic-based methods are very appealing as they directly look for phylogenetic incongruence---instead of relying on indirect evidence from pairwise similarities or site pattern frequencies---they are usually much more computationally demanding than similarity-based or substitution distribution methods. In practice, sophisticated methods such as those based on phylo-HMMs are only applicable to a small number of sequences~\citep[e.g., 4 for][]{husmeier2003detecting}, unless the space of tree topologies that can be assigned to the blocks is heavily restricted~\citep{minin2005dual,webb2008phylogenetic}. For this reason, some authors have 
proposed the use of parsimony criteria~\citep{hein1993heuristic,maydt2006recco,ruths2006recomp,ane2011detecting}, well-known to lead to faster computations at the cost of some loss in accuracy. We note in particular the \textsc{Mdl} approach by \cite{ane2011detecting}, which seeks to solve an optimization problem that is a simplification of that by \cite{hein1993heuristic} mentioned above, in that all breakpoints receive the same penalty.

The work we present here falls naturally in this context of phylogenetic- and parsimony-based methods for recombination-aware alignment partitioning. Our goal is not to present a new software tool for this task---although we do provide implementations of some novel algorithms---but to investigate a number of natural formulations of the problem, and provide exact algorithms to solve them. 
We also investigate the relative strengths of the alternative formulations.
Our optimization problems are defined in terms of the \emph{homoplasy} within each block (informally, the amount of recurrent substitution that is needed to explain the sequences in that block), although equivalent formulations could be expressed in terms of parsimony scores. This choice is motivated by the observation that deciding whether a block has a level of homoplasy below a (small) constant is
polynomially solvable~\citep{baca-lagergren,nearperfect}.

Another key choice we made is related to the well-known observation that the number of blocks inferred by many methods for alignment partitioning is sensitive to parameters such as the size of a sliding window, or the cost/penalty of breakpoints or recombinations, relative to that of substitutions, for parsimony-based methods such as \textsc{RecPars} and \textsc{Mdl}~\citep{hein1993heuristic,ane2011detecting}. To a lesser extent, this also holds for HMM-based statistical approaches, which rely on the use of priors on breakpoint frequency~\citep{husmeier2003detecting} or on the total number of breakpoints~\citep{minin2005dual}. In our problem formulations, instead of introducing a parameter that \emph{indirectly} influences the number of blocks inferred, we chose to \emph{directly} constrain the maximum number of blocks in the partition. 
We believe this has the merit of making more explicit the dependency of block partitioning on user-defined parameters.

\emph{Overview of the article.} We start by introducing necessary preliminaries in Section \ref{sec:prelim}, and then move on to defining four different parsimony-based models for detecting recombination breakpoints in Section \ref{subsec:contigDef}. Some of these models aim to minimize the maximum homoplasy in a block---or the relative frequency of homoplasy in a block---while others take an aggregate perspective, considering all blocks together. Although these models constitute fairly natural optimization criteria for parsimony-based models, one of our primary contributions is that, for all models, explicit pseudocode, rigorous proofs of correctness and detailed running time analyses are given: these are provided in Section \ref{sec:algs}.
Several of the algorithms are
Fixed Parameter Tractable (FPT) \cel{for binary characters}, meaning---in essence---that certain natural parameters of the problems have an independent, and thus limited, contribution to the overall running time~\citep[We refer the reader to][for more background on FPT algorithms]{FlumGrohe:2006,downey2012parameterized}. In Section \ref{sec:combine} we describe algorithms that attempt to \emph{merge} blocks, from a pre-computed block partition, such that parsimony-motivated criteria are optimized (these models are formally defined in \cel{Section}~\ref{subsec:combineDef}). In Section \ref{sec:experiments} we describe our publicly-available software package \textsc{CutAl}, which implements the four algorithms described in Section \ref{sec:algs}. We have performed a number of experiments on simulated data testing the ability of the four algorithms to recover the locations of breakpoints; in particular, we look at the influence of the number and length of blocks, the number of taxa and the branch lengths (of the  phylogenies used to experimentally generate alignments) on the overall performance of the algorithms.
This data has also been made publicly available. Our analysis suggests that, of our four algorithms, aggregate minimization of the total amount of homoplasy seems most effective, and is highly accurate under many variations of experimental parameters. 
In Section \ref{sec:discussion} we present our overall conclusions and
propose a number of interesting directions for future work. 

\section{Preliminaries}
\label{sec:prelim}

Let~$\cS$ be a set of character states. Throughout the paper, we assume that the cardinality~$s$ of~$\cS$ is a constant. An \emph{alignment}~$A$ is an $n\times m$ matrix with elements from~$\cS$. 
Denote by $A_j$ the $j$-th column of $A$ and denote by $A_{i,j}$ the element in the $i$-th row and $j$-th column of $A$, for any $1\leq i \leq n$ and $1\leq j \leq m$.
We also refer to the columns of an alignment as \emph{characters}.

A \emph{block}  $\contBlock{A}{i}{j}$ in $A$ is the alignment formed by the columns $A_i, \dots, A_j$, for some $1 \leq i \leq j \leq m$.
When $A$ is clear from context we will write $\contBlockS{i}{j}$ as shorthand for  $\contBlock{A}{i}{j}$. 
When the indices $i$ and $j$ are not important, we often write $B$ to denote a block. \cel{Note that blocks are always composed of contiguous columns.}
The number of columns in a block~$B$ is denoted~$|B|$. 
A \emph{block partitioning} of~$A$ is a partition 
$\cB = B_1, \dots B_b$
of the columns of~$A$, 
such that each $B_h$ is a block.

The \emph{null score} $s_0(A)$ of an alignment~$A$ is $\sum_{j=1}^m s_0(A_j)$, with $s_0(A_j)$ the number of different character states appearing in column~$A_j$ minus one. 

A \emph{phylogenetic tree} on~$X$ is an unrooted  tree with no degree-2 vertices and leaf set~$X$. An $h$-\emph{near perfect phylogeny} for an alignment~$A$ is a phylogenetic tree~$T=(V,E)$ on $\{x_1,\ldots ,x_n\}$ with a mapping~$\tau:V\rightarrow \cS^m$ such that $\tau(x_i)$ $\cel{=A_i}$ and 
\[
\sum_{\{u,v\}\in E} d_h(\tau(u),\tau(v)) \leq s_0(A) + h
\]
with~$d_h(s_1,s_2)$ the Hamming distance of~$s_1$ and~$s_2$ (the number of positions where they differ). 
The \emph{parsimony score} $PS(T,A)$ can now be defined as the minimum value of~$s_0(A)+h$ such that~$T$ is an $h$-near perfect phylogeny for~$A$. A 0-near perfect phylogeny is a \emph{perfect phylogeny}, which has parsimony score~$s_0(A)$.

An alignment $A$ is said to have \emph{homoplasy score} $h$ if $h$ is the minimum integer for which 
there exists an $h$-near perfect phylogeny for $A$. The \emph{total homoplasy} of a block partitioning is the sum of the homoplasies of its blocks. Observe that the total homoplasy of a block partition of an alignment is at most the homoplasy of the alignment.
Denote by $h(B)$ the homoplasy of a block $B$.
Denote by $r(B) = \frac{h(B)}{|B|}$ the \emph{homoplasy ratio} of $B$.

We note some established results concerning the calculation of homoplasy scores. In particular, we note that calculating the homoplasy score of an alignment is fixed-parameter tractable \cel{for binary characters}:

\begin{theorem}\citep{nearperfect}\label{thm:binaryHomFPT}
Given an~$n \times m$ alignment~$A$ with elements from~$\cS$ with $|\cS| = 2$, it can be decided in~$O(21^h+8^hnm^2)$ time whether~$A$ has an~$h$-near perfect phylogeny. 
\end{theorem}

\begin{theorem}\citep{baca-lagergren}\label{thm:nonBinaryHomXP}
Given an~$n \times m$ alignment~$A$ with elements from~$\cS$ with $|\cS| = s$, it can be decided in~$O(nm^{O(h)}2^{O(h^2s^2)})$ time whether~$A$ has an~$h$-near perfect phylogeny. 
\end{theorem}

We also recall the well-known four-gamete test~\citep{Buneman1971}, which characterizes alignments for which there exists a \cel{perfect phylogeny} 
when $|\cS| = 2$.

\medskip

\begin{theorem}[Four-gamete test]\citep{Buneman1971}
Let~$A$ be an alignment on~$\cS$ where $\cS = \{0,1\}$.
Call two characters~$j,k$ of~$A$ \emph{incompatible} if there exist rows $i_1, i_2, i_3, i_4$ such that\\ 
$A_{i_1,j} = 1 = 1= A_{i_1,k}$\\  $A_{i_2,j} = 1 \neq 0= A_{i_2,k}$\\ 
$A_{i_3,j} = 0 \neq 1= A_{i_2,k}$\\ 
$A_{i_4,j} = 0 = 0= A_{i_4,k}$.\\
Then~$A$ has homoplasy score~$0$ if and only if no two characters of~$A$ are incompatible.
\end{theorem}



A column~$A_j$ is called an \emph{uninformative site} if some character state~$s \in \cS$ appears $n- s_0(A_j)$ times in $A_j$, and every other character state appears at most once. Otherwise, $A_j$ is an \emph{informative site}.
Observe that a column $A_j$ is an uninformative site if and only if every phylogenetic tree on $\{x_1, \dots, x_n\}$ is a perfect phylogeny for $A_j$.
From a parsimony perspective, uninformative sites provide no meaningful information, and we will therefore ignore them when considering the accuracy of block partitions (see Section~\ref{sec:experiments}).
For a given block $B$ in $A$ that contains at least one informative site, 
let $B'$ be the minimal contiguous block in $B$ that contains all \cel{the} informative sites \cel{of} $B$. Then we call $B'$ the \emph{informative restriction} of $B$. Note that $B'$ can still contain uninformative sites, but the first and last columns in $B'$ are guaranteed to be informative.
A block $B$ is said to be \emph{completely uninformative} if every column in $B$ is an uninformative site.
Given a block partition $\cB$ of $A$, the \emph{informative restriction} of $B$ is derived from $B$ by deleting all completely \cel{uninformative} blocks, and replacing each remaining block with its informative restriction.

\subsection{Partitioning into contiguous blocks}
\label{subsec:contigDef}

We consider a number of block partitioning problems in which the aim is to partition an alignment into a small number of blocks with homoplasy as small as possible.

\totHomScore\\
\textbf{Given:} an alignment~$A$ and integer~$b$.\\
\textbf{Find:} a block partitioning of~$A$ into~$\cel{b'\leq } b$ blocks $B_1,\ldots ,B_{b'}$,
such that $\sum_{i = 1}^{b'} h(B_i)$ is minimized.

\maxHomRate\\
\textbf{Given:} an alignment~$A$ and integer~$b$.\\
\textbf{Find:} a block partitioning of~$A$ into~$\cel{b'\leq } b$ blocks $B_1,\ldots ,B_{b'}$,
such that $\max_{i = 1}^{b'} r(B_i)$ is minimized.

\maxHomScore\\
\textbf{Given:} an alignment~$A$ and integer~$b$.\\
\textbf{Find:} a block partitioning of~$A$ into~$\cel{b'\leq } b$ blocks $B_1,\ldots ,B_{b'}$,
such that $\max_{i = 1}^{b'} h(B_i)$ is minimized. 

\totHomRate\\
\textbf{Given:} an alignment~$A$ and integer~$b$.\\
\textbf{Find:} a block partitioning of~$A$ into~$\cel{b'\leq } b$ blocks $B_1,\ldots ,B_{b'}$,
such that $\sum_{i = 1}^{b'} r(B_i)$ is minimized. 

Note that the variant of these four problems when the blocks in the output are not required to be contiguous is NP-complete, as it was proved by~\cite{LSS2013} that the restricted case where the total homoplasy is 0 and the number of blocks is fixed to an integer $b \geq 3$ (called \textsc{$b$-Character-Compatibility}) is NP-complete.

\subsection{Combining blocks to non-contiguous blocks}
\label{subsec:combineDef}

A \emph{multiblock} for an alignment~$A$ is an alignment formed by some subset of columns in~$A$. Thus, the difference between a block and multiblock is that a multiblock is not necessarily contiguous.
We say a multiblock is a \emph{$p$-multiblock} if it can be partitioned into at most $p$ contiguous parts---i.e., it is the union of at most $p$ blocks. 
A \emph{non-contiguous block partitioning} of an alignment~$A$ is a partitioning of the columns of~$A$ into mulitblocks.
The total homoplasy of a non-contiguous block partitioning is defined similarly as for block partitionings. The next problem studied in this paper aims at merging blocks with small total homoplasy to a small number of non-contiguous blocks. It is formally defined as follows.

\textsc{Block Combining}\\
\textbf{Given:} an alignment~$A$, a block partitioning~$\cB$ of~$A$ with total homoplasy~$h$ and an integer~$c$.\\
\textbf{Decide:} does there exist a non-contiguous block partitioning with total homoplasy~$h$ that consists of at most~$c$ non-contiguous blocks, each of which is a union of blocks from~$\cB$?

\medskip

Since  allowing blocks to be completely atomized into many small contiguous parts 
seems to be too flexible, we also consider a variant where we only allow a partition into $p$-multiblocks.

 We say that two blocks~$B_j$ and~$B_k$ are \emph{mergeable} if 
their union~$B_j \cup B_k$
has homoplasy equal to the homoplasy of~$B_j$ plus the homoplasy of~$B_k$. 
In the following problem, we assume that input blocks~$B_j$ and~$B_{j+1}$ are not mergeable for any~$j$. This 
is a reasonable assumption in cases where the input block partition is derived from a block partitioning algorithm. It
holds, for example, when the blocks are those generated by Algorithm~\ref{alg:fewparblock}, and it also holds for solutions to \totHomScore{} provided the value of $b$ is chosen to be as small as possible without increasing~$h$.

\textsc{Block Combining to $p$-Multiblocks}\\
\textbf{Given:} an alignment~$A$, a block partitioning~$\cB$ of~$A$ with total homoplasy~$h$ such that $B_j$ and $B_{j+1}$ are not mergeable for any $j$, and integer~$c$.\\
\textbf{Decide:} does there exist a non-contiguous block partitioning with total homoplasy~$h$ that consists at most~$c$ $p$-multiblocks, each of which is a union of blocks from~$\cB$?

%

\section{Block partitioning algorithms}
\label{sec:algs}

\subsection{Minimizing total homoplasy}

In this section, we describe an approach that can be used to solve the \totHomScore{} problem.
The key observation is that, given a suitable method to calculate the homoplasy score of each possible block in~$A$, an optimal block partitioning can be found using standard dynamic programming techniques.
Algorithm~\ref{alg:tothompart} describes this dynamic programming technique, under the assumption that a value~$\phi(B)$ has been calculated for each possible block~$B$.
For the purposes of solving \totHomScore, we let~$\phi(B)$ be the homoplasy score of block~$B$.
However, by changing the function~$\phi$ we can also use Algorithm~\ref{alg:tothompart} to solve other block partitioning problems.

In particular, for the implementation described in Section~\ref{sec:experiments}, we will not use exact homoplasy scores for $\phi$, but instead use the values returned by the heuristic parsimony solver \textsc{Parsimonator} (\url{https://sco.h-its.org/exelixis/web/software/parsimonator/index.html}).

We begin by proving the correctness of Algorithm~\ref{alg:tothompart} with respect to an arbitrary function~$\phi$.

\medskip

\begin{lemma}\label{totDPCorrect}
 Given an alignment~$A$ with columns~$A_1,\ldots ,A_m$, an integer~$b$,
 and a value~$\phi(B)$ for each block~$B$ in~$A$,
 Algorithm~\ref{alg:tothompart} returns the minimum value $h$ for which there exists a block partitioning $B_1,\ldots ,B_{b'}$ of~$A$ into~$b' \leq b$ blocks such that $\sum_{i=1}^{b'} \phi(B_i) = h$, or $\infty$ if no such block partitioning exists.
\end{lemma}
\begin{proof}
 For each triplet of integers $(i,j,b')$ such that $1 \leq i \leq j \leq m$ and $1 \leq b' \leq b$,  Algorithm~\ref{alg:tothompart} calculates a value~$h_{part}(i,j,b')$. We first claim that for each choice of $(i,j,b')$, $h_{part}(i,j,b')$ is equal to the minimum $h$ for which there exists a block partitioning $B_1,\ldots ,B_{b'}$ of $\contBlock{A}{1}{j}$ with  $\sum_{i'=1}^{b'} \phi(B_{i'}) = h$, such that $B_{b'} = \contBlock{A}{i}{j}$, i.e., the last block consists of columns $A_i$ to $A_j$, and that $h_{part}(i,j,b') = \infty$ if no such block partitioning exists.
 
 We prove this claim by induction on $j$. We start with the base case, $j = 1$. In this case, the only possible partitioning on $\contBlock{A}{1}{j}$ is the one consisting of the single block $B_1 = \contBlockS{1}{1}$. Thus $h_{part}(i,j,b')$ should be equal to $\phi(\contBlockS{1}{j})$ if $i = b' = 1$, and $\infty$ otherwise. It can be seen that this is the value calculated by Algorithm~\ref{alg:tothompart}, and thus the claim is correct for $j = 1$.
 
 Now suppose that $j > 1$.
 If $i = 1$, then the only possible block partitioning is the one consisting of the single block $B_1 = \contBlockS{1}{j}$. Thus $h_{part}(i,j,b')$ should be equal to $\phi(\contBlockS{1}{j})$ if $i = b' = 1$, and $\infty$ otherwise. Again this is the value calculated by Algorithm~\ref{alg:tothompart}.
 If $i > 1$, then we have that an optimal block partitioning consists of a $b'-1$-block partitioning for $\contBlock{A}{1}{(i-1)}$ together with the block $\contBlockS{i}{j}$. 
 Moreover, the last block in the block partitioning of $\contBlock{A}{1}{(i-1)}$ must be $\contBlockS{i'}{(i-1)}$ for some $i' \leq i-1$.
 Thus in the case $j > 1, i > 1$, the total value $\sum_{i=1}^{b'} \phi(B_i)$
 for an optimal block partitioning $B_1,\ldots ,B_{b'}$ of~$\contBlockS{1}{j}$ is equal to $h_{part}(i',i-1,b'-1) + \phi(\contBlockS{i}{j})$, for the choice of $i'$ that minimizes this value. As this is exactly what the algorithm calculates, the claim is correct. This completes the inductive proof.
 
 It remains to observe that any block partitioning of~$A$ into at most $b$ blocks must have exactly $b'$ blocks for some $1 \leq b' \leq b$, and the last block must be $\contBlock{A}{i}{m}$ for some $1 \leq i \leq m$.
 It follows that the value of an optimal block partitioning can be found by taking the minimum value of $h_{part}(i,m,b')$ for all choices of $i$ and $b'$. \qed
\end{proof}

The following Lemma is clear from the structure of Algorithm~\ref{alg:tothompart} and is stated without proof.

\medskip

\begin{lemma}
 Given an alignment~$A$ with columns~$A_1,\ldots ,A_m$, an integer~$b$,
 and a value~$\phi(B)$ for each block~$B$ in~$A$, 
 Algorithm~\ref{alg:tothompart} has running time $O(bm^3)$.
\end{lemma}

Although we do not give a full proof here, we observe that Algorithm~\ref{alg:tothompart} can easily be converted into an algorithm that returns a block partitioning $B_1,\ldots ,B_{b'}$ of~$A$ into~$b' \leq b$ blocks that minimizes $\sum_{i=1}^{b'} \phi(B_i)$. Indeed, an optimal block partitioning for $\contBlock{A}{1}{j}$ with $b'$ blocks can be constructed by finding the value $i$ for which $h_{part}(i,j,b')$ is minimized, recursively finding an optimal block partitioning for $\contBlock{A}{1}{(i-1)}$ with $b'-1$ blocks (if $i >1$), and combining this block partitioning with the block $B_{b'} = \contBlock{A}{i}{j}$.
We therefore have the following lemma.

\medskip
\begin{lemma}\label{lem:totDPConstructive}
 Let $\phi$ be a function on blocks of $A$ such that the value of $\phi(B)$ can be calculated in $f(n,m)$ time for any block $B$, and let $b$ be an integer.
 Then in $f(n,m)m^2 + O(bm^3)$ time, we can calculate a block partitioning $B_1,\ldots ,B_{b'}$ of~$A$ into~$b' \leq b$ blocks such that $\sum_{i=1}^{b'} \phi(B_i)$ is minimized.
\end{lemma}

The $m^2$ factor comes from the need to calculate $\phi(B)$ for each of the $O(m^2)$ blocks $B$ in $A$.

Using Theorems~\ref{thm:binaryHomFPT} and~\ref{thm:nonBinaryHomXP}
we can now prove the following theorem.

\medskip


\begin{theorem}
For binary characters,
the \totHomScore{} problem 
can be solved in time 
$O(h21^h+8^hhnm^4)$,
\markj{where $h$ is is the maximum homoplasy of a block in the block partitioning,}
and is thus
fixed-parameter tractable \markj{with respect to $h$.}
For $s$-state characters, the problem can be solved in time 
$O(hnm^{O(h)}2^{O(h^2s^2)})$. 
\end{theorem}
\begin{proof}
 Recall that by Theorem~\ref{thm:binaryHomFPT}, for binary characters it can be decided in $O(21^h+8^hnm^2)$ 
 whether an alignment has an $h$-near perfect phylogeny~\citep{nearperfect}. It follows that given an integer $h$, the homoplasy score of a block can be found in  $O(h21^h+8^hhnm^2)$ time if this score is at most $h$.
Similarly by Theorem~\ref{thm:nonBinaryHomXP}, for $s$-state characters, the homoplasy score of a block can be  found in  $O(hnm^{O(h)}2^{O(h^2s^2)})$ time if this  score is at most $h$~\citep{baca-lagergren}.

So now, given an integer $h$ let $\phi$ be the function on blocks in $A$ such that $\phi(B)$ is equal to the homoplasy score of block $B$ if this is at most $h$, and $\phi(B) = \infty$ otherwise.
It remains to apply \cel{Lemma}~\ref{lem:totDPConstructive} using this function, with $f(n,m) = O(h21^h+8^hhnm^2)$ for binary characters and  $f(n,m) = O(hnm^{O(h)}2^{O(h^2s^2)})$ for $s$-state characters.
For binary characters, this gives a running time of 
$O(h21^h+8^hhnm^2)m^2 + O(bm^3) = O(h21^h+8^hhnm^4)$ (as we may assume $b \leq m$), 
and for $s$-state characters a running time of  $O(hnm^{O(h)}2^{O(h^2s^2)})m^2 + O(bm^3) = O(hnm^{O(h)}2^{O(h^2s^2)})$  (as $O(m^{O(h)})\cdot m^2 + O(bm^3) = O(m^{O(h) + 2 + 4}) = O(m^{O(h)}$). \qed
\end{proof}

We also observe that by letting~$\phi(B)$ be the homoplasy ratio of a block $B$,  Algorithm~\ref{alg:tothompart} can be used to solve \totHomRate.

\medskip

\begin{algorithm}
 \KwData{Alignment~$A$ with columns~$A_1,\ldots ,A_m$, integer~$b$,
 a value~$\phi(B)$ for each block~$B$ in~$A$ (for instance, $\phi(B)$ is the homoplasy score of $B$).}
 \KwResult{Minimum value $h$ for which there exists a block partitioning $B_1,\ldots ,B_{b'}$ of~$A$ into~$b' \leq b$ blocks such that $\sum_{i=1}^{b'} \phi(B_i) = h$, or $\infty$ if no such block partitioning exists.}
 \For{$j=1,\ldots ,m$}{
 $h_{part}(1,j,1) := \phi(\contBlockS{1}{j})$;\\
 	\For{$i=2,\ldots ,j$}{
    	$h_{part}(i,j,1) := \infty$;
    }
 	\For{$b'=2,\ldots ,b$} {
 	
        $h_{part}(1,j,b') = \infty$\;
        \For{$i=2,\ldots ,j-1$}{
            $h_{part}(i,j,b') := \min_{1 \leq i' \leq i-1} h_{part}(i',i-1,b'-1) + \phi(\contBlockS{i}{j} )$;\\
        }
    }
 }
 \Return{$\min_{b' \leq b, i\leq m}h_{part}(i,m,b')$}
 \caption{Algorithm {\sc ToHoPar}($A,b$).
 \label{alg:tothompart}}
\end{algorithm}

\subsection{Minimizing  homoplasy ratio per block}


In this section, we describe an approach that can be used to solve the \maxHomRate{} problem.
Similarly to \totHomScore{}, the key observation is that after calculating the homoplasy score (and thus the homoplasy ratio) of each possible block in~$A$, an optimal block partitioning can be found using standard dynamic programming techniques.
Algorithm~\ref{alg:homratepart} describes this dynamic programming technique, under the assumption that a value~$\phi(B)$ has been calculated for each possible block~$B$.
For the purposes of solving \maxHomRate, we let~$\phi(B)$ be the homoplasy score of block~$B$. 

\medskip

\begin{algorithm}
 \KwData{Alignment~$A$ with columns~$A_1,\ldots ,A_m$, integer~$b$,
 a value~$\phi(B)$ for each block~$B$ in~$A$ (for instance, $\phi(B)$ is the homoplasy ratio of $B$).}
 \KwResult{Minimum value $r$ for which there exists a block partitioning $B_1,\ldots ,B_{b'}$ of~$A$ into~$b' \leq b$ blocks such that $\max_{i=1}^{b'} \phi(B_i) = r$, or $\infty$ if no such block partitioning exists.}
 \For{$j=1,\ldots ,m$}{
 $r_{part}(1,j,1) := \phi(\contBlockS{1}{j})$;\\
 	\For{$i=2,\ldots ,j$}{
    	$r_{part}(i,j,1) := \infty$;
    }
 	\For{$b'=2,\ldots ,b$} {
 	
        $r_{part}(1,j,b') = \infty$\;
        \For{$i=2,\ldots ,j-1$}{
            $r_{part}(i,j,b') := \max(\phi(\contBlockS{i}{j}), \min_{1 \leq i' \leq i-1} r_{part}(i',i-1,b'-1) )$;\\
        }
    }
 }
 \Return{$\min_{b' \leq b, i\leq m}r_{part}(i,m,b')$}
 \caption{Algorithm {\sc HoRaPar}($A,b$).
 \label{alg:homratepart}}
\end{algorithm}


\begin{lemma}\label{maxDPCorrect}
 Given an alignment~$A$ with columns~$A_1,\ldots ,A_m$, an integer~$b$,
 and a value~$\phi(B)$ for each block~$B$ in~$A$,
 Algorithm~\ref{alg:homratepart} returns the minimum value $r$ for which there exists a block partitioning $B_1,\ldots ,B_{b'}$ of~$A$ into~$b' \leq b$ blocks such that $\max_{i=1}^{b'} \phi(B_i) = r$, or $\infty$ if no such block partitioning exists.
\end{lemma}
\begin{proof}
Observe that Algorithm~\ref{alg:homratepart} is identical to Algorithm~\ref{alg:tothompart}, except for line 9 which handles construction of $r_{part}(i,j,b')$ in the case where $b'>1$ and $j>i>1$. Consequently, the proof of this lemma is identical to that of Lemma~\ref{totDPCorrect}, except for the case $b'>1$ and $j>i>1$, and we omit the other cases.

As with Lemma~\ref{totDPCorrect}, we prove by induction on $j$ that $r_{part}(i,j,b')$ is equal to the minimum $r$ for which there exists a block partitioning $B_1,\ldots ,B_{b'}$ of $\contBlock{A}{1}{j}$ with  $\max_{i'=1}^{b'} \phi(B_{i'}) = r$, such that $B_{b'} = \contBlock{A}{i}{j}$, and that $r_{part}(i,j,b') = \infty$ if no such block partitioning exists.
For the case  $b'>1$ and $j>i>1$,
 we have that an optimal block partitioning consists of a $b'-1$-block partitioning for $\contBlock{A}{1}{(i-1)}$ together with the block $\contBlockS{i}{j}$. 
 Moreover, the last block in the block partitioning of $\contBlock{A}{1}{(i-1)}$ must be $\contBlockS{i'}{(i-1)}$ for some $i' \leq i-1$.
 Thus the value $\max_{i=1}^{b'} \phi(B_i)$
 for an optimal block partitioning $B_1,\ldots ,B_{b'}$ of~$\contBlockS{1}{j}$ is equal to the maximum of $r_{part}(i',i-1,b'-1)$ and $\phi(\contBlockS{i}{j})$, for the choice of $i'$ that minimizes this value. As this is exactly what the algorithm calculates, the claim is correct, and we have completed the proof for this case. \qed
\end{proof}

The following Lemma is clear from the structure of Algorithm~\ref{alg:homratepart} and is stated without proof.

\begin{lemma}
 Given an alignment~$A$ with columns~$A_1,\ldots ,A_m$, an integer~$b$,
 and a value~$\phi(B)$ for each block~$B$ in~$A$, 
 Algorithm~\ref{alg:homratepart} has running time $O(bm^3)$.
\end{lemma}

As with Algorithm~\ref{alg:tothompart}, we observe that Algorithm~\ref{alg:homratepart} can be made to return a block partitioning using standard backtracking techniques, and that using existing homoplasy algorithms, a certain parameterization of \maxHomRate{} for binary characters is fixed-parameter tractable.


\begin{theorem}
For binary characters,
the \maxHomRate{} problem 
can be solved in time 
$O(h21^h+8^hhnm^4)$,
\markj{where $h$ is is the maximum homoplasy of a block in the block partitioning,}
and is thus
fixed-parameter tractable \markj{with respect to $h$.}
For $s$-state characters, the problem can be solved in time 
$O(hnm^{O(h)}2^{O(h^2s^2)})$. 
\end{theorem}

We also observe that by letting~$\phi(B)$ be the homoplasy score of a block $B$,  Algorithm~\ref{alg:tothompart} can be used to solve \maxHomScore.



\subsection{Minimizing number of blocks}

In this section, we consider a variation of the \maxHomScore{} problem, described below.

\fewParBlo\\
\textbf{Given:} an alignment~$A$ and integer~$h$.\\
\textbf{Find:} a block partitioning~$B_1,\ldots ,B_{b}$, of~$A$ into a minimum number of blocks 
such that each block $B_k$ admits an $h$-near perfect phylogeny.

Algorithm~\ref{alg:fewparblock} solves the \fewParBlo{} problem.

\begin{algorithm}[]
 \KwData{Alignment~$A$ with columns~$A_1,\ldots ,A_m$ and an integer~$h$.}
 \KwResult{Block partitioning $B_1,\ldots ,B_b$ of~$A$ into a minimum number of blocks, such that each block~$B_k$ admits an $h$-near perfect phylogeny.}
 
 $b:=1$\;
 $j:=1$\;
 \For{$i=2,\ldots ,m$}{
 	\If{ 
 	$\contBlock{A}{j}{i}$
 	does not admit an $h$-near perfect phylogeny}{
        $B_b := \contBlock{A}{j}{(i-1)}$\;    
        $b:=b+1$\;
        $j:=i$
    }
 }
 $B_b := \contBlock{A}{j}{m}$\;
 \Return{$B_1,\ldots ,B_b$}
 \caption{Algorithm {\sc FewParBlo}($A,h$)}
 \label{alg:fewparblock}
\end{algorithm}
\medskip

\begin{theorem}
Algorithm~\ref{alg:fewparblock} solves the \fewParBlo{} problem.
\end{theorem}
\begin{proof}
Let~$B_1,\ldots ,B_b$ be a block partitioning produced by Algorithm~\ref{alg:fewparblock}, and let~$F_1,\ldots ,F_c$ be a block partitioning into a minimum number of blocks such that the index of the first column where $B_1,\ldots ,B_b$ and $F_1,\ldots ,F_c$ differ (i.e., the minimum $i$ such that columnn $A_i$ appears in blocks $B_k$ and $F_{k'}$ for some $k \neq k'$) is as large as possible.
Consider the smallest~$k$ for which~$F_k\neq B_k$. It is not possible that~$F_k$ has more columns than~$B_k$ since otherwise the algorithm would have extended~$B_k$ with another column. Hence,~$F_k$ has fewer columns than~$B_k$. However, then we can add a column to~$F_k$ and obtain a solution with the same number of blocks as $F_1,\ldots ,F_c$ where the first columns where it differs from $B_1,\ldots ,B_b$ is one larger. This contradicts the assumption that the first column where $B_1,\ldots ,B_b$ and $F_1,\ldots ,F_c$ differ is as large as possible. Hence, we conclude that $B_1,\ldots ,B_b$ is identical to $F_1,\ldots ,F_c$ and therefore optimal. \qed
\end{proof}

The running time of Algorithm~\ref{alg:fewparblock} is~$m$ times the running time of the $h$-near perfect phylogeny algorithm, hence $O(21^hm+8^hnm^3)$ for binary characters and $O(nm^{O(h)}2^{O(h^2s^2)})$ for general $s$-state characters. 
Consequently, the \fewParBlo{} problem is fixed-parameter tractable \cel{for binary characters}. 


\begin{corollary}
\fewParBlo{} is fixed-parameter tractable when the parameter is~$h$ \cel{for binary characters}.
\end{corollary}

If the number of blocks $b$ is known to be small ($b\leq m / \log m$), the following algorithm may also be useful. It uses binary search to find the location of the recombination site between two blocks.

\begin{algorithm}[]
 \KwData{Alignment~$A$ with columns~$A_1,\ldots ,A_m$ and an integer~$h$.}
 \KwResult{Block partitioning $B_1,\ldots ,B_k$ of~$A$ into a minimum number of blocks, such that each block~$B_i$ admits an $h$-near perfect phylogeny.}
 $\text{end\_previous\_block}:=0$\; 
 \While{$\text{end\_previous\_block}<m$}{
    $l = \text{end\_previous\_block}+1$\;
    $u_2=m$\; 
    \While{$l < u_2$}
    {
        $u_1=\ceil{(l+u_2)/2}$\;
	    \uIf{
	    $\contBlock{A}{\text{(end\_previous\_block}+1)}{u_1}$
	    admits an $h$-near perfect phylogeny}{
            $l:=u_1$\;
            
    	}\Else{
           $u_2:=u_1-1$\;
        }
    }
    $B_b:=\contBlock{A}{(\text{end\_previous\_block}+1)}{u_2}$\;
    $b:=b+1$\;
    $\text{end\_previous\_block}:=u_2$\;
 }
\Return{null\;}
 \caption{Algorithm {\sc FewParBlo2}($A,h,b$)}
 \label{alg:homoperblockv2}
\end{algorithm}



The correctness of this algorithm follows from a similar argument as for the previous algorithm: this algorithm uses binary search instead of linear search to find the longest possible block, using $l$  and $u_2$ as lower and upper bounds on the last column of that block. \textcolor{black}{The only extra observation needed is the following. If $\contBlock{A}{i}{j}$ admits an $h$-near perfect phylogeny, then so does $\contBlock{A}{i}{(j-1)}$; and, similarly, if $\contBlock{A}{i}{j}$ does not admit an $h$-near perfect phylogeny, then neither does $\contBlock{A}{i}{(j+1)}$.}

Note that the running time of the part within the while loop is dominated by the function checking whether an $h$-near perfect phylogeny exists. As this while loop performs a binary search on a list of length $m-l$, its contents are executed $\log_2(m-l)<\log_2(m)$ times. Finally, this while loop is contained in a for loop which runs $b$ times at most. Therefore the running time is bounded by $b\log_2(m)$ times the running time of the $h$-near perfect phylogeny algorithm, which theoretically gives a factor $m/(b\log_2 m)$ improvement over the previous algorithm. An important note here is that actual running time may also depend on implementation of the $h$-near perfect phylogeny algorithm. If the worst case for this algorithm is only attained in NO cases, the improvement of the second algorithm may be less than expected as the first algorithm encounters exactly one NO case per block, and the second may encounter more such cases.

\section{Combining blocks to non-contiguous blocks}
\label{sec:combine}

We now consider the \textsc{Block Combining to $p$-Multiblocks} problem and first show this problem to be NP-complete for~$p\geq 3$ by reduction from the \textsc{Bounded Coloring} problem: given a graph~$G$, does there exist a coloring of the vertices of~$G$ with at most~$c$ colors such that each color is used at most~$p$ times and adjacent vertices always get different colors? This problem is NP-complete for each fixed~$p\geq 3$~\citep{boundedcolorings}.

\medskip

\begin{theorem}
For every integer $p \geq 3$, the \textsc{Block Combining to $p$-Multiblocks} problem is NP-complete for~$h=0$.
\end{theorem}
\begin{proof}
Given an instance of the \textsc{Bounded Coloring} problem, that is an integer $p \geq 3$ and a graph $G = (V,E)$ with $n = |V|$, we build the following blocks $B^j$ of $3n+1$ aligned sequences $S_0, S_1, \ldots, S_{3n}$ of length $n$ for each vertex $v_j$ of $G$, illustrated in Figure~\ref{fig:reductionFromColoring}:
\begin{itemize}
\item $B^j_{0,k} = 0$ for all $k \in \{1,\ldots,n\}$ (that is, the first row in $B^j$ consists of all $0$'s);
\item 
 $B^j_{3j-2,j} = B^j_{3j-1,j} = 1$ and $B^j_{k,j} = 0$ for all $k$ in $\{1,\ldots ,3j-3\} \cup \{3j,\ldots ,n\}$;
\item for all $i \in [1,\ldots ,n]$ distinct from $j$ such that $v_j$ adjacent with $v_i$, $B^j_{3i-1,i} = B^j_{3i,i} = 1$ and $B^j_{k,i} = 0$ for all $k \in \{1,\ldots ,3i-2\} \cup \{3i+1, \ldots ,n\}$;
\item for all $i \in \{1,\ldots ,n\}$ distinct from $j$ such that $v_j$ not adjacent with $v_i$, $B^j_{3i-2,i} = B^j_{3i-1,i} = B^j_{3i,i} = 1$ and $B^j_{k,i} = 0$ for all $k \in \{1, \ldots ,3i-3\} \cup \{3i+1, \ldots , n\}$.
\end{itemize}

\begin{figure}
\centering
\begin{tabular}{lcccc}
~ & $B^1$ & $B^2$ & $B^3$ & $B^4$\\
$S_0$ & 0000 & 0000 & 0000 & 0000\\
$S_1$ & 1000 & 1000 & 0000 & 0000\\
$S_2$ & 1000 & 1000 & 1000 & 1000\\
$S_3$ & 0000 & 1000 & 1000 & 1000\\
$S_4$ & 0100 & 0100 & 0100 & 0100\\
$S_5$ & 0100 & 0100 & 0100 & 0100\\
$S_6$ & 0100 & 0000 & 0100 & 0100\\
$S_7$ & 0000 & 0010 & 0010 & 0010\\
$S_8$ & 0010 & 0010 & 0010 & 0010\\
$S_9$ & 0010 & 0010 & 0000 & 0010\\
$S_{10}$ & 0000 & 0001 & 0001 & 0001\\
$S_{11}$ & 0001 & 0001 & 0001 & 0001\\
$S_{12}$ & 0001 & 0001 & 0001 & 0000\\
\end{tabular}
\caption{\label{fig:reductionFromColoring} An instance of the \textsc{Block Combining  to $p$-Multiblocks} problem built from an instance $G = \{\{v_1,$ $v_2,$ $v_3,$ $v_4\},$ $\{v_1v_3,$ $v_1v_4\}\}$ of the \textsc{Bounded Coloring} problem with at most $p$ vertices per color.}
\end{figure}


Recall from the four-gamete test that an alignment $A$ has homoplasy 0 if and only if no two characters are incompatible, where characters $j,k$ are incompatible if there 
exist rows $i_1, i_2, i_3, i_4$ such that
$A_{i_1,j}  = 1= A_{i_1,k}$, $A_{i_2,j} = 1 \neq 0= A_{i_2,k}$, $A_{i_3,j} = 0 \neq 1= A_{i_2,k}$ and 
$A_{i_4,j} = 0= A_{i_4,k}$.

First note that by construction, the only characters which may be
equal to 1 in each block $B_j$ are the $i$-th character of
the block on sequences $S_{3i-2}$, $S_{3i-1}$ and $3i$, all other
characters are equal to 0. Therefore, there is at most one character
equal to 1 in each line of block $B_j$, therefore $B_j$ contains
no incompatible characters, so the block $B_j$ has homoplasy 0,
therefore we have built a proper instance of the
\textsc{Block Combining to $p$-multiblocks} problem \cel{for $h=0$}.

Now, suppose that $G$ is a positive instance of the 
\textsc{Bounded Coloring} problem using at most $c$ colors used at most $p$ times, then there exist $c$ independent
sets $I_{1}, \ldots, I_{c}$ in $G$. We claim that the corresponding
$p$-multiblocks have homoplasy 0.

To prove this claim, let us consider 2 characters $c$ and $c'$
in two distinct blocks $B$ and $B'$ in the same $p$-multiblock,
such that $c$ is the $i$-th character of block $B$ and
$c'$ is the $i'$-th character of block $B'$. If $i \neq i'$,
the two characters are compatible because the 1s in those two
characters do not appear in the same line. Otherwise,
by construction the only
1s may appear in the sequences $S_{3i-2}$, $S_{3i-1}$ and $S_{3i}$
for these 2 characters. The vertices $v$ and $v'$ corresponding
to blocks $B$ and $B'$ are not adjacent because they are part
of an independent set, therefore by construction the sequences
$S_{3i-2}$, $S_{3i-1}$ and $S_{3i}$ all contain 1 for one of these
two characters, therefore both characters are compatible. So in all cases each $p$-multiblock has homoplasy 0.

To prove the reverse, let us assume that it is possible to
merge the blocks of the instance of the \textsc{Block Combining to $p$-Multiblocks}
we have built so that each $p$-multiblock has homoplasy 0.
For each $p$-multiblock $P$, let us consider
two vertices $v_i$ and $v_{i'}$ of $G$ corresponding
to two blocks $B_i$ and $B_{i'}$ of $P$.
Suppose by contradiction that $v_i$ and $v_{i'}$ are adjacent.
Then we have $B^{i}_{3i-1,i}=1 = B^{i'}_{3i-1,i}$, $B^{i}_{3i-2,i} = 1 \neq 0 = B^{i'}_{3i-2,i}$, $B^{i}_{3i,i} = 0 \neq 1 = B^{i'}_{3i,i}$ and $B^{i}_{0,i}=0=B^{i'}_{0,i}$.
Thus the $i$-th character of block $B^{i}$ and of block $B^{i'}$ are
incompatible, and therefore the homoplasy cost of the block combination
containing $B^{i}$ and $B^{i'}$ is strictly greater than 0; a contradiction.
Therefore, no pair of vertices corresponding to the blocks of the block
combination are adjacent, thus each set of vertices corresponding to
each of the $c$ block combinations is an independent set, so
$G$ has a $c$-coloring.

Thus, we have built an instance of the \textsc{Block Combining to $p$-Multibocks} problem where a merge into $c$ non-contiguous blocks (each containing at most $p$ contiguous blocks) with total homoplasy~0 for each of these $c$ blocks is possible if and only if there is a $c$-coloring of~$G$ where each color is used at most $p$ times, therefore \textsc{Block Combining to $p$-Multiblocks} is NP-hard.

Given a merge into $c$ non-contiguous blocks, it is easy to check if all the characters of each of these blocks are compatible, and that each non-contiguous block contains at most $p$ contiguous blocks, so \textsc{Block Combining to $p$-Multiblocks} is NP-complete. \qed
\end{proof}




We now focus on the case~$p=2$, i.e., the problem \textsc{Block Combining to 2-Multiblocks}. 


\begin{algorithm}[]
 \KwData{Alignment~$A$, a block partitioning~$\cB$ of~$A$ with total homoplasy~$h$ such that $B_j$ and $B_{j+1}$ are not mergeable for any $j$, and integer~$c$.}
 \KwResult{At most~$c$ 2-multiblocks, each of which is a union of blocks from~$\cB$, such that the total homoplasy of the 2-multiblocks is~$h$ (if such a solution exists).}
 Construct a graph~$G=(V,E)$ with
 a vertex for each block and
 an edge~$\{B_j,B_k\}$ if~$B_j$ and~$B_k$ are mergeable;\\
 Find a maximum matching~$M$ in~$G$;\\
 \For{each edge $\{B_j,B_k\}\in M$} {
 merge~$B_j$ and~$B_k$ into a 2-multiblock;
 }
 \For{each vertex~$B_j$ that is not covered by~$M$} {make~$B_j$ a 2-multiblock consisting of a single continuous part;}
 \uIf{the number of obtained 2-multiblocks is at most~$c$} {\Return{the obtained 2-multiblocks;}} 
 \Else{\Return{null;}}
 \caption{Algorithm {\sc BloCo2Mul}($A,\cB,h,c$)}
 \label{alg:2multiblocks}
\end{algorithm}

\medskip

We now argue correctness of the algorithm. Suppose that there exist~$c'$ 2-multiblocks with total homoplasy~$h$ and such that each of the 2-multiblocks is a union of blocks from~$\cB$. Each 2-multiblock that consists of exactly two contiguous parts corresponds to an edge of~$G$. Let~$M'$ be the set of all such edges and observe that they form a matching in~$G$. Each 2-multiblock that consists of a single contiguous part corresponds to a vertex of~$G$. Moreover, each vertex of~$G$ that is not covered by~$M'$ corresponds to such a 2-multiblock. Hence, $c'=(|V|-2|M'|)  + |M'|= |V|-|M'|$. It follows that we get the smallest possible number~$c'$ of 2-multiblocks by choosing a maximum cardinality matching~$M'$, which Algorithm~\ref{alg:2multiblocks} does.

The following theorem follows directly, since it can be checked in polynomial time whether an alignment has homoplasy~0~\citep{perfectphylogeny,fastperfectphylogeny} and for binary states this is fixed-parameter tractable with parameter~$h$~\citep{nearperfect}.

\medskip

\begin{theorem}
\textsc{Block Combining to $p$-Multiblocks} can be solved in polynomial time for~$p=2$ and~$h=0$ and it is fixed-parameter tractable with parameter~$h$ for binary characters and~$p=2$.
\end{theorem}

\medskip

We now continue to the general \textsc{Block Combining to $p$-Multiblocks} problem with~$p\geq 2$. 
This problem can be solved by Algorithm~\ref{alg:independentset}. We say that a collection of blocks~$\cB'\subseteq\cB$ is \emph{mergeable} if the total homoplasy of~$\cB'$ is equal to the homoplasy of the alignment obtained by combining all blocks from~$\cB'$. Note that, in general, it may happen that a set of blocks is not mergeable even if the blocks in the set are pairwise mergeable.

\begin{algorithm}[]
 \KwData{
 Alignment~$A$, a block partitioning~$\cB$ of~$A$ with total homoplasy~$h$ such that $B_j$ and $B_{j+1}$ are not mergeable for any $j$, and integer~$c$.}
 \KwResult{At most~$c$ $p$-multiblocks, each of which is a union of blocks from~$\cB$, such that the total homoplasy of the $p$-multiblocks is~$h$.}
 Construct a graph~$G=(V,E)$ with a vertex for each set~$\cB'\subseteq\cB$ with~$|\cB'|\leq p$ that is mergeable and an edge~$\{\cB',\cB''\}$ if~$\cB'\cap\cB''\neq\emptyset$\;
 Give each vertex~$\cB'$ a weight equal to~$|\cB'|-1$ \cel{(i.e., the number of blocks in $\cB'$ minus one)}\;
 Find a maximum weight independent set~$I$ in~$G$\;
\For{each input block $B\in\cB$ not in any element of $I$}{add the vertex $\{B\}$ to $I$\;}
\For{each vertex~$\cB'\in I$} {
 	Merge the blocks in~$\cB'$ into a $p$-multiblock\;
 }
 \uIf{the number of obtained $p$-multiblocks is at most~$c$}{\Return{the obtained $p$-multiblocks}} 
 \Else{\Return{null}}
 \caption{Algorithm {\sc BloCopMul}($A,\cB,h,c$)}
 \label{alg:independentset}
\end{algorithm}

Correctness of Algorithm~\ref{alg:independentset} follows from the following argument. First we argue that each $p$-multiblock partitioning of $\cB$ with total homoplasy $h$ and $c$ multiblocks corresponds to an independent set in $G$ containing all input blocks $B\in\cB$ of weight $|\cB|-c$. Then we show that each independent set $I$ of $G$ with weight $|\cB|-c$ gives a $p$-multiblock partitioning of $\cB$ in $c$ multiblocks. We now prove this in detail.

\medskip

\begin{lemma}
Algorithm~\ref{alg:independentset} is correct.
\end{lemma}
\begin{proof}
Suppose we have a $p$-multiblock partitioning $\{\cB'_i\}_{i\in[c]}$ of $\cB$ with total homoplasy $h$ into $c$ multiblocks, then each multiblock must consist of a mergeable set of input blocks. Indeed, the total homoplasy $h(\cB'_i)$ of each multiblock $\cB'_i=\{B_{i_1},\dots,B_{i_{j_i}}\}$ is at least the sum $\sum_{k=1}^{j_i}h(B_{i_k})$ of the homoplasies of the contained blocks. So, if the total homoplasy of all multiblocks $\sum_{i=1}^ch(\cB'_i)$ is at most $h$, and the total homoplasy $\sum_{B\in\cB}h(B)$ of $\cB$ is also $h$, none of the multiblocks may have strictly greater total homoplasy than the sum of the homoplasy of the contained blocks, i.e., $h(\cB'_i)>\sum_{k=1}^{i_j}h(B_{i_k})$ is not allowed. Hence $h(\cB'_i)=\sum_{k=1}^{i_j}h(B_{i_k})$ for each multiblock $\cB'_i$, or in other words, each multiblock~$\cB'_i$ is mergeable. 

This means that every multiblock corresponds to a vertex of $G$ and because the multiblocks form a partition of $\cB$, there are no edges between these vertices. Hence the nodes corresponding to the chosen multiblocks form an independent set in $G$. The weight of this independent set is $w(I)=\sum_{i=1}^c |\cB'_i|-1=|\cB|-c$.

Now define $\cB_I:=\cup_{\cB'\in I} \cB'$ and suppose we find an independent set $I$ with $\cB_I \neq \cB$, then $\cB_I$ must be a strict subset of $\cB$ as each element of $I$ is a subset of $\cB$. Let $B\in \cB$ be an input block not chosen for any element of the independent set (i.e., $B\not \in\cB_I$). Then $\{B\}$ is a vertex of $G$ because the total homoplasy of $\{B\}$ is trivially equal to the homoplasy of $B$. Furthermore, there is no edge $\{\{B\},\cB'\}$ for any element $\cB'\in I$. This means that adding $\{B\}$ to $I$ gives a new independent set, and its weight is $w(I)+w(\{B\})=w(I)+|\{B\}|-1=w(I)$. As the algorithm uses the same procedure to add elements to an independent set until each block is in one of the elements, we may assume that each input block is in at least one element of the independent set $I$, i.e., $\cB_I = \cB$.

Because there is an edge between two vertices in $G$ exactly if they are not disjoint, an independent set of $G$ corresponds to a partition of (a subset of) $\cB$, so each input block is in at most one element of the independent set. We conclude that each block of the input is in exactly one element of an independent set produced by the algorithm. 

Now we look at the weight of such an independent set in $G$. Let $I=\{\cB'_i\}_{i\in[k]}$ be an independent set in $G$, then the weight $w(I)$ of $I$ is $w(I)=\sum_{i=1}^k|\cB_i'|-1= |\cB|-k$. This means that if the weight of $I$ is $|\cB|-c$, then the number of multiblocks in the partitioning corresponding to $I$ is $c$. Hence, there is a solution with at most $c$ $p$-multiblocks if and only if there is an independent set in~$G$ of weight at least $|B|-c$. Because the algorithm finds a maximum weight independent set in~$G$, it outputs a solution with the minimum number of $p$-multiblocks. \qed
\end{proof}

\section{Experiments}\label{sec:experiments}

We have implemented the four optimization models \maxHomScore, \totHomScore, {\maxHomRate} and \totHomRate{} in the open-source software package \textsc{CutAl}, which can be downloaded from \url{https://github.com/celinescornavacca/CUTAL}. This program does not use the near-perfect phylogeny algorithms from~\citep{nearperfect} because no implementation of these algorithms is available and they are not expected to run efficiently for larger data sets. Instead, for the construction of phylogenetic trees we implemented a brute-force algorithm for small data sets and use the heuristic parsimony solver \textsc{Parsimonator} (\url{https://sco.h-its.org/exelixis/web/software/parsimonator/index.html}) for larger data sets. \textsc{Parsimonator} is used by the high-performance software RAxML-light (and more recently, ExaML~\citep{kozlov2015examl}) to warm-start the search for maximum likelihood trees.

We performed 2 experiments,
whose main goal was to act as a first ``sanity check'' to pick up potential problems of the inferences produced by the 4 models---under ideal conditions---rather than to assess their behavior on realistic data.
The first experiment concerned alignments with 2 blocks, the second experiment concerned multiple blocks (ranging from 3 to 6). Full output of the experiments can be downloaded from the \textsc{CutAl} GitHub page. A detailed description of the experimental protocol, and some brief information concerning running times, can be found in the appendix. To enhance readability we describe here only the overall structure of the experiments.\\
\\
The high-level, informal idea is to generate an alignment consisting of 400 nucleotides and $k \in \{2,3,4,5,6\}$ blocks, such that the locations of the breakpoints, \steven{denoted \emph{breakpoints}, are specified as a parameter of the experiment (in the case of the 2-block experiment) or are chosen randomly (in the case of the multiple-block experiment).}
For each block, a random phylogenetic tree is chosen, where all the trees have $t \in \{5,10,20,50\}$ taxa and all branches of the tree have length $b\ell \in \{0.001, 0.01, 0.1\}$. These branch lengths have been chosen to mimic the three more common branch length categories in the OrthoMaM database \citep{orthomam}. 
We use \textsc{Dawg}  \citep{cartwright2005dna} to simulate
a DNA alignment corresponding to these parameters. The alignment is fed to \textsc{CutAl} and optimal solutions
\steven{with $k$ blocks}
under the four models \maxHomScore, \totHomScore, {\maxHomRate} and \totHomRate{} are computed. (\textsc{CutAl} computes the optima for all four models in a single execution, so the four models are always applied to exactly the same input data.)
For each model, we assess how far the breakpoints chosen by \textsc{CutAl} are from the experimentally generated breakpoints. To measure this, we use the \emph{breakpoint error}, defined as the number of informative sites separating the inferred breakpoints from the correct ones. 
To decrease the impact of randomness, we run each ($t, b\ell, breakpoints, k$) 
combination 20 times (i.e., we obtain 20 `replicates'), taking the average and standard deviation of the breakpoint errors.\\
\\
\steven{As stated above, t}he only difference between the 2-block and the multiple-block experiment is the way \emph{breakpoints} are dealt with: in the 2-block experiment ($k=2$) the exact \emph{location} of the single internal breakpoint is controlled experimentally, so a given combination of experimental parameters---corresponding to a single row of Table \ref{tab:2block}---is more accurately summarized as ($t, b\ell, location$). In the multiple-block experiment only the \emph{number} of blocks is controlled experimentally, and the breakpoints themselves are selected randomly: so a given combination of experimental parameters, corresponding to a single row of Table \ref{tab:multiblock}, is in this case actually ($t, b\ell, k$).\\
\\
The results are summarized in Table \ref{tab:2block} (for the 2-block experiment) and Table \ref{tab:multiblock} (for the multiple-block experiment.) 
For each parameter combination, we report
average and standard deviation of the breakpoint error (ranging over the 20 replicates). In each row of the tables, the smallest average (ranging across the four optimization models) is shown in bold. At the foot of each table we provide average-of-averages and average-of-standard-deviations.\\



\subsection{Analysis of experiments}

Both experiments indicate that, overall,
\totHomScore{} is the best algorithm, then \maxHomRate, then \maxHomScore, and finally \totHomRate{}. The average-of-averages 
shown at the foot of each table emphasize this. A number of observations can be made:
\begin{itemize}

\item Across both experiments, and across most parameter combinations, \totHomScore{} achieves an average error across the 20 replicates of (much) less than 1.0. This means that \totHomScore{} is almost always inferring breakpoint locations that are 
within one informative site of the their correct locations.
In the 2-block experiment, the only exception to this is several of the $(t, b\ell) = (5, 0.1)$ parameter combinations, which have slightly higher averages and standard deviations. (These combinations also disrupt the other optimization models.) We believe this is because trees with long branches and few taxa produce very little phylogenetic signal
enabling the separation of regions generated by different trees.
In the multiple-block experiment, the $(5, 0.1)$ combinations behave similarly, but there $(50, 0.001)$ combinations also cause the average to rise marginally above 1.0.


\item In both experiments, many of the parameter combinations cause both \maxHomScore{} and \totHomRate{} to suffer very large breakpoint errors.  Interestingly, in both experiments these huge errors only occur when branch lengths are long $(0.1)$.
We reflect below on the reasons for this, which are different for the two models.


\item Overall, the performance of \maxHomRate{} seems somewhat correlated with \totHomScore{}. Unlike the other two optimization models, 
neither of these models produces any spectacularly large breakpoint errors.

\item In both experiments, for extremely short branch lengths $(0.001)$ all models exhibit similar performance. This is 
connected to the fact that, for such extremely short branch lengths, an alignment typically contains very few informative sites.
To give a concrete example: in the 2-block experiment, 
for $b\ell=0.001$, an alignment generated with $t=5$
might contain only 1 informative site per block, rising to at most 15-20 for $t=50$. 
Given the definition of breakpoint error this means that there are many sites that do not contribute to the error, so the breakpoint error will be limited in magnitude irrespective of the exact optimisation criterion being used.


   
   
\item 
Close observation of the full experimental results, available on the \textsc{CutAl} GitHub page, for \maxHomScore{} shows that this criterion tends to underestimate the size of the largest block. In the 2-block experiments, this is particularly evident for $b\ell=0.1$, where for $location = 50, 100, 150$ the inferred breakpoint
is nearly always to the right of the correct position (e.g., for $t=50$ it is always between 150 and 205, irrespective of the true position of the breakpoint). This phenomenon is also detectable for $t=50$, $b\ell=0.01$ and $location = 50$, where the inferred breakpoint is on average more than 10 informative sites to the right of the correct beakpoint (cf.~Table \ref{tab:2block}). This behavior is not surprising, as the goal of this problem is to minimize the homoplasy in the block with the most homoplasy (which usually coincides with the largest block).
Thus, solutions of \maxHomScore{} will tend to
reduce the size of the largest block to reduce its homoplasy (while increasing the homoplasy of the blocks flanking the largest block).
The incentive to do so is stronger when the homoplasy within the largest block is much bigger than the homoplasy in its flanking blocks, a difference that we expect to be particularly pronounced when the tree has very long (or many) branches (for large values of $b\ell$ and/or $t$) or when the blocks have very different sizes.
As expected, when the true blocks have similar sizes (e.g., for $location=200$ in the 2-block experiment) this phenomenon has a tendency to disappear (cf.~Table \ref{tab:2block}).

\item 
\totHomRate{} suffers from a problem that can be seen as a very severe inverse of that of \maxHomScore{}: whenever per-site substitutions are frequent (in our experiments for $b\ell=0.1$), the partitions it produces tend to consist of a single large block and a number of smaller blocks. For example, for the extreme case where $t=50$ and $b\ell=0.1$: (a) in the 2-block experiments, all of the partitions returned consist of a block of 1 or 2 sites, and a block containing the remaining sites; (b) in the multi-block experiments, all partitions consist of a number of blocks of 1 or 2 sites, and a single block with all the remaining sites. (2-site blocks are much rarer than 1-site blocks, as they only occur in <5\% of these partitions.)
The reason why this happens is that when substitutions are frequent\----in other words, when the sites are saturated with substitutions---\-each block will have a high homoplasy ratio, regardless of whether the MP tree for that block accurately describes its evolution, or not. The only exception to this are very small blocks consisting of few sites, in which case it becomes possible to find a tree that explains those few sites with little or no homoplasy (e.g., 1-site blocks always have homoplasy ratio 0). 
So, in order to minimize the sum of homoplasy ratios for $k$ blocks, it becomes convenient to have $k-1$ small blocks with very small homoplasy ratios and one large block with high homoplasy ratio, rather than to have a realistic block partition where each block has an anyway high homoplasy ratio.
This is a very serious issue for \totHomRate{}, which results in very high breakpoint errors whenever the alignment is saturated with substitutions.

\end{itemize}


\begin{table}
\footnotesize 
\centering
\begin{tabular}{|l|l|l||l|l||l|l||l|l||l|l|}
\hline
 &  &  & \makecell{MAX \\ HOM \\ SCORE} &  & \makecell{TOT \\ HOM\\ SCORE} &  & \makecell{MAX \\ HOM \\ RATIO} &  & \makecell{TOT \\ HOM \\ RATIO} &  \\ \hline
taxa & \makecell{branch \\ length} & \makecell{breakpoint\\ location} & avg & sd & avg & sd & avg & sd & avg & sd \\ \hline
5 & 0.001 & 50 & \bf{0.00} & 0.00 & \bf{0.00} & 0.00 & \bf{0.00} & 0.00 & \bf{0.00} & 0.00 \\ \hline
5 & 0.001 & 100 & \bf{0.05} & 0.22 & \bf{0.05} & 0.22 & \bf{0.05} & 0.22 & \bf{0.05} & 0.22 \\ \hline
5 & 0.001 & 150 & \bf{0.15} & 0.37 & \bf{0.15} & 0.37 & \bf{0.15} & 0.37 & \bf{0.15} & 0.37 \\ \hline
5 & 0.001 & 200 & \bf{0.15} & 0.49 & \bf{0.15} & 0.49 & \bf{0.15} & 0.49 & \bf{0.15} & 0.49 \\ \hline
5 & 0.01 & 50 & \bf{0.25} & 0.64 & \bf{0.25} & 0.64 & \bf{0.25} & 0.64 & \bf{0.25} & 0.64 \\ \hline
5 & 0.01 & 100 & \bf{0.35} & 0.67 & \bf{0.35} & 0.67 & 0.75 & 1.45 & 0.75 & 1.45 \\ \hline
5 & 0.01 & 150 & 0.75 & 1.33 & \bf{0.70} & 1.30 & 0.85 & 1.50 & 0.85 & 1.50 \\ \hline
5 & 0.01 & 200 & 1.20 & 2.26 & \bf{1.15} & 2.28 & 1.25 & 2.29 & 1.25 & 2.29 \\ \hline
5 & 0.1 & 50 & 13.70 & 10.58 & \bf{5.75} & 16.28 & 15.50 & 26.66 & 18.80 & 26.82 \\ \hline
5 & 0.1 & 100 & 4.85 & 8.51 & \bf{1.85} & 4.46 & 3.50 & 5.28 & 15.80 & 19.66 \\ \hline
5 & 0.1 & 150 & 1.15 & 1.73 & \bf{0.70} & 1.22 & 1.20 & 1.47 & 8.25 & 11.67 \\ \hline
5 & 0.1 & 200 & 2.45 & 2.31 & \bf{2.30} & 6.61 & 3.90 & 8.28 & 10.05 & 15.42 \\ \hline
10 & 0.001 & 50 & \bf{0.05} & 0.22 & \bf{0.05} & 0.22 & \bf{0.05} & 0.22 & \bf{0.05} & 0.22 \\ \hline
10 & 0.001 & 100 & \bf{0.05} & 0.22 & \bf{0.05} & 0.22 & \bf{0.05} & 0.22 & \bf{0.05} & 0.22 \\ \hline
10 & 0.001 & 150 & \bf{0.10} & 0.31 & \bf{0.10} & 0.31 & \bf{0.10} & 0.31 & \bf{0.10} & 0.31 \\ \hline
10 & 0.001 & 200 & \bf{0.10} & 0.31 & \bf{0.10} & 0.31 & \bf{0.10} & 0.31 & \bf{0.10} & 0.31 \\ \hline
10 & 0.01 & 50 & 0.15 & 0.37 & \bf{0.05} & 0.22 & 0.15 & 0.49 & 0.20 & 0.70 \\ \hline
10 & 0.01 & 100 & 0.25 & 0.72 & \bf{0.00} & 0.00 & 0.05 & 0.22 & 0.05 & 0.22 \\ \hline
10 & 0.01 & 150 & 0.25 & 0.64 & \bf{0.10} & 0.45 & 0.25 & 0.64 & 0.25 & 0.55 \\ \hline
10 & 0.01 & 200 & 0.15 & 0.49 & \bf{0.00} & 0.00 & 0.10 & 0.45 & 0.05 & 0.22 \\ \hline
10 & 0.1 & 50 & 48.45 & 11.49 & \bf{0.15} & 0.37 & 3.00 & 4.70 & 33.40 & 36.38 \\ \hline
10 & 0.1 & 100 & 22.25 & 4.41 & \bf{0.35} & 0.59 & 3.30 & 5.03 & 77.00 & 55.78 \\ \hline
10 & 0.1 & 150 & 9.80 & 5.78 & \bf{0.45} & 0.89 & 3.00 & 3.39 & 63.40 & 51.56 \\ \hline
10 & 0.1 & 200 & 3.95 & 3.43 & \bf{0.40} & 0.60 & 4.20 & 4.50 & 34.30 & 53.41 \\ \hline
20 & 0.001 & 50 & \bf{0.15} & 0.49 & \bf{0.15} & 0.49 & \bf{0.15} & 0.49 & \bf{0.15} & 0.49 \\ \hline
20 & 0.001 & 100 & \bf{0.35} & 0.67 & \bf{0.35} & 0.67 & \bf{0.35} & 0.67 & \bf{0.35} & 0.67 \\ \hline
20 & 0.001 & 150 & 0.15 & 0.37 & \bf{0.10} & 0.31 & \bf{0.10} & 0.31 & \bf{0.10} & 0.31 \\ \hline
20 & 0.001 & 200 & \bf{0.35} & 0.49 & \bf{0.35} & 0.49 & \bf{0.35} & 0.49 & \bf{0.35} & 0.49 \\ \hline
20 & 0.01 & 50 & 1.70 & 4.17 & \bf{0.00} & 0.00 & 0.80 & 1.85 & 1.05 & 2.06 \\ \hline
20 & 0.01 & 100 & 0.35 & 0.99 & \bf{0.05} & 0.22 & 0.25 & 0.72 & 0.15 & 0.37 \\ \hline
20 & 0.01 & 150 & 0.05 & 0.22 & \bf{0.00} & 0.00 & 0.10 & 0.31 & 0.10 & 0.31 \\ \hline
20 & 0.01 & 200 & 0.40 & 0.60 & \bf{0.10} & 0.31 & \bf{0.10} & 0.31 & \bf{0.10} & 0.31 \\ \hline
20 & 0.1 & 50 & 91.40 & 10.45 & \bf{0.20} & 0.41 & 4.15 & 5.78 & 81.10 & 94.81 \\ \hline
20 & 0.1 & 100 & 44.75 & 5.74 & \bf{0.00} & 0.00 & 5.15 & 5.76 & 94.65 & 38.34 \\ \hline
20 & 0.1 & 150 & 20.00 & 3.37 & \bf{0.10} & 0.31 & 2.25 & 2.45 & 150.75 & 38.21 \\ \hline
20 & 0.1 & 200 & 4.15 & 3.13 & \bf{0.25} & 0.44 & 3.65 & 3.88 & 175.55 & 5.17 \\ \hline
50 & 0.001 & 50 & \bf{0.00} & 0.00 & \bf{0.00} & 0.00 & \bf{0.00} & 0.00 & \bf{0.00} & 0.00 \\ \hline
50 & 0.001 & 100 & \bf{0.05} & 0.22 & \bf{0.05} & 0.22 & \bf{0.05} & 0.22 & \bf{0.05} & 0.22 \\ \hline
50 & 0.001 & 150 & 0.30 & 0.57 & \bf{0.25} & 0.55 & 0.35 & 0.59 & 0.30 & 0.57 \\ \hline
50 & 0.001 & 200 & 0.30 & 0.73 & 0.25 & 0.55 & \bf{0.15} & 0.37 & 0.20 & 0.41 \\ \hline
50 & 0.01 & 50 & 11.30 & 7.26 & \bf{0.00} & 0.00 & 1.35 & 2.28 & 4.65 & 5.71 \\ \hline
50 & 0.01 & 100 & 4.00 & 3.49 & \bf{0.10} & 0.31 & 0.55 & 1.43 & 0.30 & 0.92 \\ \hline
50 & 0.01 & 150 & 2.45 & 2.72 & \bf{0.00} & 0.00 & 0.65 & 1.04 & 0.30 & 0.57 \\ \hline
50 & 0.01 & 200 & 1.05 & 1.47 & \bf{0.00} & 0.00 & 0.10 & 0.31 & 0.10 & 0.31 \\ \hline
50 & 0.1 & 50 & 113.45 & 5.60 & \bf{0.05} & 0.22 & 2.05 & 4.15 & 49.00 & 0.00 \\ \hline
50 & 0.1 & 100 & 57.80 & 3.11 & \bf{0.05} & 0.22 & 1.30 & 1.69 & 98.90 & 0.31 \\ \hline
50 & 0.1 & 150 & 28.25 & 2.99 & \bf{0.00} & 0.00 & 2.60 & 3.12 & 168.90 & 41.09 \\ \hline
50 & 0.1 & 200 & 2.15 & 1.42 & \bf{0.00} & 0.00 & 2.50 & 2.14 & 198.80 & 0.41 \\ \hline
 & & & & & & & & & &\\ \hline
 \hline
 & & AVG: & 10.32 & 2.45 & 0.37 & 0.93 & 1.48 & 2.28 & 26.90 & 10.68 \\ \hline
\end{tabular}
\caption{2-block experiment}
\label{tab:2block}
\end{table}

\begin{table}
\centering
\footnotesize
\begin{tabular}{|l|l|l||l|l||l|l||l|l||l|l|}
\hline
 &  &  & \makecell{MAX \\ HOM \\ SCORE} &  & \makecell{TOT \\ HOM \\ SCORE} & & \makecell{MAX \\ HOM \\ RATIO} &  & \makecell{TOT \\ HOM \\ RATIO} &  \\ \hline
taxa & \makecell{branch \\ length} & blocks & avg & sd & avg & sd & avg & sd & avg & sd \\ \hline
5 & 0.001 & 3 & \bf{0.03} & 0.11 & \bf{0.03} & 0.11 & \bf{0.03} & 0.11 & \bf{0.03} & 0.11 \\ \hline
5 & 0.001 & 4 & \bf{0.30} & 0.44 & \bf{0.30} & 0.44 & \bf{0.30} & 0.44 & \bf{0.30} & 0.44 \\ \hline
5 & 0.001 & 5 & *** & *** & *** & *** & *** & *** & *** & *** \\ \hline
5 & 0.001 & 6 & *** & *** & *** & *** & *** & *** & *** & *** \\ \hline
5 & 0.01 & 3 & 0.80 & 1.22 & \bf{0.60} & 1.12 & 0.68 & 1.13 & \bf{0.60} & 1.12 \\ \hline
5 & 0.01 & 4 & 1.33 & 1.34 & 1.30 & 1.28 & \bf{1.20} & 1.08 & \bf{1.20} & 1.08 \\ \hline
5 & 0.01 & 5 & 0.93 & 1.16 & \bf{0.88} & 1.18 & \bf{0.88} & 1.18 & \bf{0.88} & 1.18 \\ \hline
5 & 0.01 & 6 & \bf{1.31} & 1.27 & \bf{1.31} & 1.27 & \bf{1.31} & 1.27 & \bf{1.31} & 1.27 \\ \hline
5 & 0.1 & 3 & 7.58 & 8.30 & \bf{3.88} & 6.02 & 6.38 & 9.38 & 10.25 & 10.47 \\ \hline
5 & 0.1 & 4 & 4.13 & 4.63 & \bf{2.93} & 3.99 & 6.85 & 8.07 & 13.40 & 7.97 \\ \hline
5 & 0.1 & 5 & 5.93 & 4.31 & \bf{4.24} & 5.97 & 6.10 & 6.19 & 23.23 & 13.64 \\ \hline
5 & 0.1 & 6 & 4.23 & 3.37 & \bf{3.38} & 4.05 & 8.25 & 5.57 & 23.60 & 9.18 \\ \hline
10 & 0.001 & 3 & \bf{0.50} & 1.08 & \bf{0.50} & 1.08 & \bf{0.50} & 1.08 & \bf{0.50} & 1.08 \\ \hline
10 & 0.001 & 4 & \bf{0.25} & 0.36 & \bf{0.25} & 0.36 & \bf{0.25} & 0.36 & \bf{0.25} & 0.36 \\ \hline
10 & 0.001 & 5 & \bf{0.56} & 0.76 & \bf{0.56} & 0.76 & \bf{0.56} & 0.76 & \bf{0.56} & 0.76 \\ \hline
10 & 0.001 & 6 & \bf{0.37} & 0.42 & \bf{0.37} & 0.42 & \bf{0.37} & 0.42 & \bf{0.37} & 0.42 \\ \hline
10 & 0.01 & 3 & 0.83 & 1.32 & \bf{0.30} & 0.91 & 1.18 & 2.89 & 1.05 & 2.86 \\ \hline
10 & 0.01 & 4 & 1.53 & 2.83 & 0.72 & 2.09 & \bf{0.42} & 0.66 & \bf{0.42} & 0.67 \\ \hline
10 & 0.01 & 5 & 0.41 & 0.46 & \bf{0.15} & 0.19 & 1.10 & 1.54 & 1.00 & 1.73 \\ \hline
10 & 0.01 & 6 & 1.28 & 1.66 & \bf{0.63} & 1.05 & 1.17 & 1.78 & 0.64 & 0.97 \\ \hline
10 & 0.1 & 3 & 17.03 & 15.03 & \bf{0.63} & 0.93 & 5.13 & 6.71 & 80.03 & 33.37 \\ \hline
10 & 0.1 & 4 & 17.62 & 11.88 & \bf{0.55} & 0.60 & 4.40 & 3.31 & 86.60 & 32.55 \\ \hline
10 & 0.1 & 5 & 14.33 & 8.27 & \bf{0.34} & 0.34 & 6.54 & 11.32 & 80.95 & 21.78 \\ \hline
10 & 0.1 & 6 & 7.27 & 5.41 & \bf{0.33} & 0.31 & 6.35 & 7.04 & 96.84 & 21.16 \\ \hline
20 & 0.001 & 3 & \bf{0.80} & 1.37 & \bf{0.80} & 1.37 & 0.85 & 1.36 & 0.85 & 1.36 \\ \hline
20 & 0.001 & 4 & 0.77 & 0.91 & 0.77 & 0.91 & \bf{0.72} & 0.91 & \bf{0.72} & 0.91 \\ \hline
20 & 0.001 & 5 & \bf{0.88} & 0.76 & \bf{0.88} & 0.76 & \bf{0.88} & 0.76 & \bf{0.88} & 0.76 \\ \hline
20 & 0.001 & 6 & \bf{0.66} & 0.60 & \bf{0.66} & 0.60 & \bf{0.66} & 0.60 & \bf{0.66} & 0.60 \\ \hline
20 & 0.01 & 3 & 1.13 & 1.19 & \bf{0.10} & 0.21 & 0.70 & 0.80 & 0.45 & 0.69 \\ \hline
20 & 0.01 & 4 & 1.37 & 1.18 & \bf{0.17} & 0.28 & 1.27 & 3.97 & 0.60 & 1.24 \\ \hline
20 & 0.01 & 5 & 1.14 & 1.65 & \bf{0.18} & 0.18 & 1.23 & 3.00 & 1.50 & 3.86 \\ \hline
20 & 0.01 & 6 & 0.63 & 0.61 & \bf{0.11} & 0.15 & 0.72 & 1.31 & 0.45 & 0.87 \\ \hline
20 & 0.1 & 3 & 39.03 & 20.64 & \bf{0.28} & 0.34 & 1.78 & 2.11 & 121.48 & 35.40 \\ \hline
20 & 0.1 & 4 & 24.80 & 11.97 & \bf{0.08} & 0.15 & 2.78 & 1.64 & 144.58 & 42.29 \\ \hline
20 & 0.1 & 5 & 18.70 & 9.01 & \bf{0.14} & 0.19 & 4.08 & 6.74 & 148.66 & 37.15 \\ \hline
20 & 0.1 & 6 & 13.41 & 7.02 & \bf{0.09} & 0.12 & 2.83 & 1.72 & 131.72 & 35.70 \\ \hline
50 & 0.001 & 3 & 0.83 & 1.40 & \bf{0.80} & 1.41 & \bf{0.80} & 1.41 & \bf{0.80} & 1.41 \\ \hline
50 & 0.001 & 4 & 1.65 & 1.83 & 1.35 & 1.59 & \bf{1.28} & 1.59 & \bf{1.28} & 1.59 \\ \hline
50 & 0.001 & 5 & 1.60 & 1.86 & 1.59 & 1.87 & 1.56 & 1.88 & \bf{1.55} & 1.89 \\ \hline
50 & 0.001 & 6 & 1.70 & 1.68 & 1.66 & 1.71 & 1.52 & 1.40 & \bf{1.48} & 1.41 \\ \hline
50 & 0.01 & 3 & 6.88 & 5.85 & \bf{0.13} & 0.22 & 1.05 & 1.16 & 1.55 & 3.69 \\ \hline
50 & 0.01 & 4 & 5.57 & 6.48 & \bf{0.15} & 0.20 & 0.73 & 0.68 & 0.80 & 1.45 \\ \hline
50 & 0.01 & 5 & 2.76 & 2.51 & \bf{0.13} & 0.21 & 1.03 & 0.89 & 1.90 & 6.83 \\ \hline
50 & 0.01 & 6 & 1.44 & 0.68 & \bf{0.09} & 0.14 & 0.91 & 1.00 & 0.19 & 0.21 \\ \hline
50 & 0.1 & 3 & 52.08 & 26.90 & \bf{0.13} & 0.22 & 3.20 & 3.37 & 135.93 & 56.92 \\ \hline
50 & 0.1 & 4 & 29.05 & 18.99 & \bf{0.07} & 0.14 & 1.07 & 0.83 & 162.22 & 44.56 \\ \hline
50 & 0.1 & 5 & 23.43 & 10.96 & \bf{0.05} & 0.10 & 1.66 & 1.04 & 149.60 & 39.10 \\ \hline
50 & 0.1 & 6 & 18.12 & 7.31 & \bf{0.10} & 0.15 & 2.04 & 1.41 & 149.65 & 34.72 \\ \hline
 & & & & & & & & & & \\ \hline
 \hline
 & & AVG & 7.32 & 4.76 & 0.75 & 1.04 & 2.07 & 2.48 & 34.42 & 11.28 \\ \hline
\end{tabular}
\caption{Multiple-block experiment}
\label{tab:multiblock}
\end{table}



\section{Discussion}
\label{sec:discussion}


{\color{black} Inference methods estimating quantities that are 
paramount to our understanding of molecular evolution (e.g. positive selection  and branch length estimates) can be seriously mislead by recombination.
Since no theory yet exists to define the amount of recombination that can be borne by
these methods to still produce the correct answer, a safer way to deal with the problem is to restrict one's attention to the \fabio{analysis} of non-recombinant loci.}

In this article, we investigated several different parsimony-based approaches towards detecting recombination breakpoints in a multiple sequence alignment, and we described algorithms for each of our formulations. These algorithms have been implemented and, via tests on simulated data, we have identified 
a number of important weaknesses of two of these approaches.
Our experiments also suggest that the best formulation of block partitioning, among those that we have considered, is \totHomScore{}. We note that this formulation 
is closely linked to that of \textsc{Mdl} \citep{ane2011detecting}: it is easy to see that for any value of the penalty parameter against the number of blocks in \textsc{Mdl}, there exists a value of the parameter $b$ for \totHomScore{} that would yield the same block partition. 



The work in this paper has led to several interesting open computational problems which could be studied in further research. 
First of all, a fundamental question is whether one can decide in polynomial time whether two blocks are mergeable, i.e., whether they can be combined into a single block without increasing the total homoplasy. Simply computing the homoplasy for each block and the combined block does not work (unless P $=$ NP) because computing the homoplasy of a block is NP-hard. Hence, a more advanced strategy would need to be developed. 
A slightly different but equally interesting question is the following. Suppose we are given a tree with minimum homoplasy for the first~$j$ columns of an alignment. Can we decide in polynomial time whether the same tree has minimum homoplasy for the first~$j+1$ columns?
We would also like to reiterate here a well-known open problem in the area, which is strongly related to the studied problems. Does there exists an FPT algorithm deciding whether there exists a tree with homoplasy at most~$h$ for a given alignment of nonbinary characters, when~$h$ is the parameter? In particular, this problem (the $h$-near perfect phylogeny problem) is even open for $3$-state characters. Answering this question positively would also extend the FPT results in this paper to nonbinary characters.

The dependency of our algorithms on the length of the alignment is at least cubic. This is due to the fact that we allow breakpoints to occur at any position of the alignment, and that we want to detect these positions precisely. These requirements could be relaxed by decreasing
the granularity of breakpoint detection, only allowing breakpoints to be inferred at
given positions of the alignments. 
This approach would make it possible to analyse longer alignments, and has already been adopted by tools such \textsc{Mdl} (which only allowed the inference of a breakpoint every 300 sites in their study~\citep{ane2011detecting}) and \textsc{Gard} (which only allows breakpoints to be inferred up to the nearest variable site~\citep{kosakovsky2006automated}).


Another avenue of research would be to adapt \textsc{CutAl}'s approach  to a likelihood-based context. For example, instead of minimizing the sum of the block homoplasies, we could aim to
maximize the sum of the $\log$ likelihoods for the blocks (basically replacing homoplasy with the $- \log$ of the likelihood in \totHomScore). Since our implementations are based on a precalculation of the homoplasy scores for contiguous subalignments of the input alignment, the same could be done for $\log$ likelihoods relatively easily, although at the cost of some efficiency. Such an approach could benefit likelihood-based programs for breakpoint detection, for example  \textsc{Gard}~\citep{kosakovsky2006automated,kosakovsky2006gard} which is based on a genetic algorithm to search the space of block partitionings, and on a information-theory penalty to limit the number of blocks. Our results imply that  partitions that \textsc{Gard} would consider optimal for a fixed number of blocks can be calculated efficiently (assuming that log-likelihoods are precalculated for a large set of intervals in the input alignment). This would eliminate the need for a heuristic search within the space of block partitionings of a fixed size, thus largely simplifying \textsc{Gard}'s optimization.


Note that while the drawbacks of maximum parsimony for phylogenetic reconstruction are known and well-understood~\citep{felsenstein1978}, an investigation of parsimony's flaws when the goal is alignment partitioning rather than tree inference is still lacking. 
In fact some authors have suggested that a possible bias of maximum parsimony in this context would be to detect extra breakpoints~\citep{ane2011detecting}, 
which is not an excessively disruptive bias, and could be resolved by merging contiguous blocks separated by breakpoints where the evidence of recombination is not sufficiently strong.

\bibliographystyle{spbasic}

\bibliography{main}

\pagebreak

\appendix
\section{Appendix: full description of experiments}

All experiments were conducted on
the Linux Subsystem (Ubuntu 16.04.6 LTS), running under Windows 10, on a 64-bit HP Envy Laptop 13-ad0xx (quad-core i7-7500 @ 2.7 GHz), with 8 Gig of memory. \textsc{CutAl} was compiled with GCC. We incorporated code from \textsc{Parsimonator} version 1.0.2. We used \textsc{DAWG} version 1.2 \citep{cartwright2005dna} to simulate alignments. This software is capable of simulating the evolution of nucleotides along the branches of a phylogenetic tree, such that different blocks of the generated alignment correspond to different phylogenetic trees.\\


\textbf{Protocol common to both experiments}
\begin{enumerate}

\item All random sampling in the experiments is with-replacement (i.e., repetition is allowed).
\item Binary phylogenetic trees, one per block, are generated by randomly splitting the taxa into two subgroups, and iterating this process on the subgroups until sets comprising single taxa are obtained (i.e., the leaves). In other words binary rooted trees are sampled following the standard Yule-Harding distribution \citep{harding1971probabilities}, which also arises in the context of the coalescent model \citep{aldous1996probability,kingman1982genealogy}.

\item The generated DNA alignments consist of 400 nucleotides, without indels.
\item When executing \textsc{CutAl}, an upper bound on the number of blocks has to be indicated. We set this equal to the number of blocks $k$.
We also feed \textsc{CutAl} $breakpoints$ so that it can compute breakpoint errors for us; this auxiliary functionality is built into \textsc{CutAl}).
\item DAWG is asked to generate a DNA alignment with $k$ blocks on $t$ taxa, where the breakpoint locations are indicated by $breakpoints$. For each block we ask that the nucleotides evolve according to the Jukes Cantor model along a randomly selected binary tree with $t$ taxa where \emph{all} the branches have length $b\ell$. (DAWG does not select the trees itself: we generate them using the method described in point 2). 
\item For each combination of experimental parameters, we perform 20 iterations (`replicates'), and then compute the average and standard deviation of the breakpoint errors for each of the four optimization models. Note that, although all replicates share the same set of experimental parameters, binary trees and alignments are not held constant across the replicates. In other words, for each replicate we select a new set of binary trees (one per block) and generate a new alignment with DAWG.
\item We \emph{never} feed \textsc{CutAl} alignments that contain completely uninformative blocks. Specifically: if a replicate produces an alignment which has at least one completely uninformative block, we discard that and try generating a new replicate, repeating as often as necessary until we obtain an alignment in which all blocks are informative. The reasons for doing this are explained in the point below. We emphasize that it is \emph{not} due to an inherent limitation of \textsc{CutAl}, but rather an artefact of the way that we measure breakpoint errors in this experiment. Practically, discarding alignments in this way means that, to obtain 20 replicates for a given parameter combination, we might have to generate (many) more than 20 alignments. In particular, for the multiple-block experiment this means we have to omit two extreme parameter combinations: essentially, those with many blocks, short branch lengths and few taxa. The chance that all blocks are informative under these parameters is so low, that it is prohibitively difficult to randomly generate such inputs. 
\item The score used for measuring beakpoint errors is defined as follows. Recall that the left-most block always starts in the first column of the alignment, and the right-most block always finishes in the last column of the alignment, so in terms of testing the accuracy with which breakpoints are inferred, we only need to look at the $k-1$ ``internal'' breakpoints. \textbf{We take the average, ranging over the $k-1$ internal breakpoints, of the absolute difference between the location of that breakpoint given by \emph{breakpoints} and the location of the corresponding breakpoint inferred by \textsc{CutAl}.} So, for example, if there are 3 blocks with internal breakpoints at positions 50 and 150 and \textsc{CutAl} infers blocks with internal breakpoints at positions 65 and 147 respectively, the score will be $(|50-65| + |150-147|)/2 = 9$. Three points are worth noting:
\begin{enumerate}
\item The score assumes a left-to-right bijection between the experimentally generated breakpoints, and those inferred by  \textsc{CutAl}. This is why we discard alignments that contain completely uninformative blocks: they disrupt this bijection and thus lead to artificially and punitively high errors. 
\item When computing the absolute difference between the locations of two corresponding internal breakpoints, uninformative sites do not contribute to the difference, in the following sense.
If there are $x$ informative sites and $y$ uninformative sites between the two breakpoints, the absolute difference is taken to be $x$, rather than $x+y$. In particular, if \textsc{CutAl} and the experimentally generated partition place corresponding breakpoints in the same region of uninformative sites, this counts as an absolute difference of $0$, i.e., total agreement. 
This is because, from the perspective of the four optimization models, all these locations are indistinguishable.

\item When computing the score,  if \textsc{CutAl} infers that optimality is obtained (for a given optimization model) by a number of blocks strictly less than $k$, \textsc{CutAl} uses the best solution with \emph{exactly} $k$ blocks to compute the boundary errors.
\end{enumerate}

\item Since uninformative sites are indistinguishable from the perspective of the four  optimization models,  \textsc{CutAl} does not return a simple block partition but rather its informative restriction. For example, if \textsc{CutAl} would return $\contBlockS{1}{65},\contBlockS{66}{147}, \contBlockS{148}{200}$ but sites 64, 65, 66, 148 and 160 are uninformative, \textsc{CutAl} will instead return $\contBlockS{1}{63},\contBlockS{67}{147}, \contBlockS{149}{200}$.



\item To enhance the readability of the results, for each parameter combination we have only reported average and standard deviation (of the breakpoint error, ranging over the 20 replicates) in the results tables, and \steven{below we also give} a rough indication of running times. However, in the supplementary material \steven{available on the \textsc{CutAl} GitHub page} we have provided the following, more complete information for each parameter combination:
\begin{enumerate}
    \item Running times of individual replicates;


    \item The number of times that, for each of the four optimization models, the optimum inferred by \textsc{CutAl} contained strictly fewer than $k$ blocks;
    \item The number of alignments that we had to sample (above the baseline of 20) in order to obtain 20 replicates where all blocks of the alignment are informative.
\end{enumerate}
\end{enumerate}

%
\textbf{Protocol specific to 2-block experiment ($k=2$) \steven{and summary of running times}}
\begin{enumerate}
\item We select $location \in \{ 50, 100,150, 200\}$. So the alignment consists of two blocks, where the first has length $location$ and the second has length $(400 - location)$. This ensures a balance between situations where the two blocks have very different sizes, and situations where they have the same size.
\item As indicated earlier: for each $t, b\ell, location$ parameter combination we take 20 replicates, and each replicate selects new trees and alignments. So in this experiment trees and alignments are random variables. 
\item For $t \in \{5,10\}$,  \textsc{CutAl} solves each replicate to optimality (for all four optimization models combined) in time ranging from 3 seconds to 11 seconds. For $t=20$ the running time varies between 11s and 43s, and within this range increases with branch length. For $t=50$ the running time varies between 57s and 425s, again increasing inside this range with branch length. \steven{Note that these running times include the time to execute \textsc{Parsimonator} on all intervals of the input alignment.}

\item See Table \ref{tab:2block} for the results. In each row, the smallest average (ranging across the four optimization models) is shown in bold.
\end{enumerate}

\textbf{Protocol specific to multiple-block experiment \steven{and summary of running times}}
\begin{enumerate}
\item We select $k \in \{3,4,5,6\}$.
\item Unlike the 2-block experiment, where every replicate uses new trees and new alignments, in this case every replicate \emph{also} randomly selects new breakpoint locations. So in this experiment the random variables are trees, alignments and breakpoint locations.
\item Breakpoint locations are sampled uniformly at random, subject to the following rules: the length of each block is a multiple of 50, with a minimum length of 50.
\item We omit the parameter combinations where $(t, b\ell, k)$ is equal to (5,0.001, 5) and (5, 0.001, 6). This is due to technicalities associated with completely uninformative blocks, discussed earlier. Specifically, with these parameter combinations the chance of observing an alignment where all blocks are informative is extremely low. In the table of results we have used ``*'' to indicate that these two parameter combinations have been excluded.
\item Running times are of a similar magnitude, and follow a similar pattern, to the 2-block experiment. The number of taxa, and then to a lesser extent the branch lengths, and then to a lesser extent the number of blocks, are correlated with increasing times.  
\item See Table \ref{tab:multiblock} for the results. In each row, the smallest average (ranging across the four optimization models) is shown in bold.
\end{enumerate}

\end{document}